 \documentclass[final,3p,times,authoryear]{elsarticle}

\usepackage{subfigure}
\usepackage{multirow}
\usepackage{algorithmic}
\usepackage{algorithm}

\usepackage{subfigure}
\usepackage{graphicx}
\usepackage{amsfonts}

\newtheorem{definition}{Definition}[section]
\newtheorem{theorem}{Theorem}[section]

\newtheorem{proof}{Proof}[section]
\newtheorem{lemma}{Lemma}[section]
\newtheorem{example}{Example}[section]

\newtheorem{stipulation}{Stipulation}[section]

\begin{document}

\begin{frontmatter}



\title{Efficient Semi-External Breadth-First Search}


\author{Xiaolong~Wan}
\ead{wxl@hit.edu.cn}

\author{Xixian~Han\corref{cor1}}
\cortext[cor1]{Corresponding author}
\ead{hanxx@hit.edu.cn}

\affiliation{organization={School of Computer Science and Technology},
	addressline={Harbin Institute of Technology}, 
	city={Harbin},
	postcode={150001}, 
	state={Heilongjiang},
	country={China}}

%

\begin{abstract}
  Breadth-first search (BFS) is known as a basic search strategy for learning graph properties. As the scales of graph databases have increased tremendously in recent years, large-scale graphs $G$ are often disk-resident. Obtaining the BFS results of $G$ in semi-external memory model is inevitable, because the in-memory BFS algorithm has to maintain the entire $G$ in the main memory, and external BFS algorithms consume high computational costs. As a good trade-off between the internal and external memory models, semi-external memory model assumes that the main memory can at least reside a spanning tree of $G$. Nevertheless, the semi-external BFS problem is still an open issue due to its difficulty. Therefore, this paper presents a comprehensive study for processing BFS in semi-external memory model. After discussing the naive solutions based on the basic framework of semi-external graph algorithms, this paper presents an efficient algorithm, named EP-BFS, with a small minimum memory space requirement, which is an important factor for evaluating semi-external algorithms. Extensive experiments are conducted on both real and synthetic large-scale graphs, where graph WDC-2014 contains over $1.7$ billion nodes, and graph eu-2015 has over $91$ billion edges. Experimental results confirm that EP-BFS can achieve up to 10 times faster.
\end{abstract}




\end{frontmatter}



\section{Introduction}\label{sec:introduction}

Breadth-first search~(BFS) is a well-known graph traversal problem~\citep{DBLP:books/daglib/0023376}, and has various applications on directed graphs, such as reachability query~\cite{DBLP:conf/sigmod/JinHWRX10}, creating web page indexes~\cite{DBLP:journals/jksucis/AttiaAK22}, and computing shortest path or minimum spanning tree on unweighted graphs~\cite{DBLP:conf/ics/FengWZLLL24,DBLP:books/daglib/0037819}. To be specific, if BFS traverses a node $u$ of a directed graph $G$, it marks $u$ as \textit{visited} and adds all the adjacent nodes of $u$ into a first-in-first-out queue $Q$. Then, it recursively traverses the first unmarked node $v$ in $Q$. Here, we take BFS on a social graph for example, as described in Example~\ref{example:G_0}.

\begin{example}
	\label{example:G_0}
	$G_0$ is a social graph shown in Figure~\ref{fig:introduction_figure}(a), where (\romannumeral1) each node represents a person, (\romannumeral2) each edge~(a solid or dotted directed line in Figure~\ref{fig:introduction_figure}(a)) connects two nodes, and (\romannumeral3) an edge $e$ from node $u$ to node $v$ denotes that $u$ follows $v$. To verify whether a person Amy~(node $a$ in $G_0$) directly or indirectly follows another person Taylor~(node $t$ in $G_0$), one may execute BFS on $G_0$ starting from node $a$, where the search order is given in Figure~\ref{fig:introduction_figure}(b). The edges that BFS passed in $G_0$ when searching from $a$ are drawn as solid lines in Figure~\ref{fig:introduction_figure}(a), while the other edges are drawn as dotted lines.
\end{example}

The BFS problem on a directed graph $G$ can be addressed efficiently, when $G$ can fully reside in the main memory. Assuming $G$ has $n$ nodes and $m$ edges, the in-memory BFS algorithm costs $O(n+m)$ time~\cite{DBLP:books/daglib/0023376}. Nevertheless, the sizes of graphs grow continuously and substantially, and many real graphs cannot be fully accommodated in the main memory~\cite{DBLP:conf/sigmod/ZhangYQCL13} and are disk-resident. For instance, as a small snapshot of the Internet, graph clueweb12 contains more than $978$ million nodes and $42$ billion edges\footnote{http://law.di.unimi.it/webdata/clueweb12/}; graph WDC-2014 includes over $1.7$ billion nodes and $64$ billion edges~\footnote{http://www.webdatacommons.org/}.

\begin{figure}[t]
	\centering
	\renewcommand{\thesubfigure}{}
	\subfigure[(a) $G_0$]{
		\includegraphics[scale = 0.5]{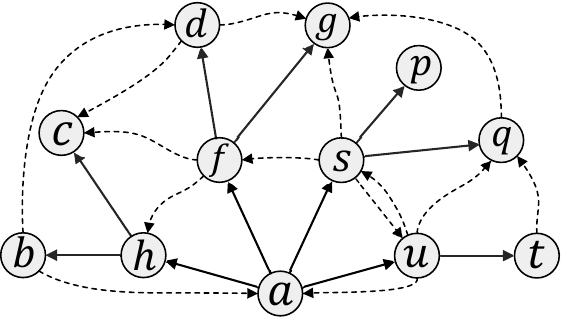}}\quad
	\subfigure[(b) Search order]{
		\includegraphics[scale = 0.5]{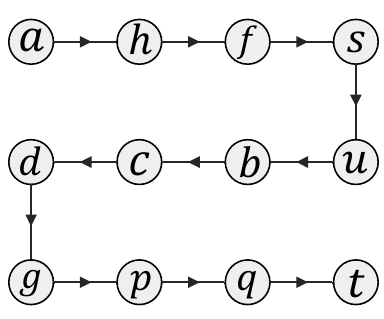}}
	
	\caption{A social graph $G_0$, and the search order of $G_0$ when BFS starts from node $a$.}
	\label{fig:introduction_figure}
\end{figure}

External memory model is proposed for processing disk-resident data, which assumes the main memory can accommodate at most $M$ elements, and $M$ can be really small. The rest, which cannot be resident in memory, is stored in external memory, and divided into different blocks. Each block contains $B$ consecutive elements. Data must be transferred to the CPU via memory to be computed. When the required data is stored only in external memory, it must be transferred from external memory to memory through an I/O operation. In an I/O operation, one disk block of data is transferred between the main memory and the disk. External BFS algorithms~\cite{DBLP:conf/soda/Meyer01,DBLP:conf/swat/ArgeT04,DBLP:conf/esa/MehlhornM02,DBLP:conf/spdp/KumarS96} aim to address the BFS problem when $M$ is very small. On general directed graphs, \textit{the best external BFS algorithm, named EM-BFS}, still relies on \textit{buffered repository tree} to guarantee its runnability when $M$ is very small~\cite{DBLP:conf/soda/BuchsbaumGVW00}. EM-BFS, due to  its dependence on the buffered repository tree, consumes high computational costs and cannot be afforded when processing large-scale graphs.


Semi-external memory model, the variant of external memory model, sets a lower bound for $M$, and $M\geq c\times n$, where $c$ is a small constant, e.g. $c=3,4,5,\dots$, instead of assuming $M$ can be arbitrarily small. It has attracted the attention of researchers for processing graph problems when the main memory cannot contain the entire $G$. For one thing, the semi-external memory model is known as a good trade-off between the internal memory model and external memory model~\cite{DBLP:conf/sigmod/ZhangYQCL13}. For another, processing graph problems in semi-external memory model is inevitable, since the commonly scale-out approaches will eventually hit the wall of technology and monetary~\cite{DBLP:conf/cidr/KerstenS17}. 

Note that, even though semi-external memory model sets a low bound on the memory size, processing graph problems in semi-external memory model is still non-trivial, and \textit{designing efficient semi-external BFS algorithm on general directed graphs $G$ is still an open issue}. The difficulty of processing graph problems in semi-external memory model comes from the huge difference between the size of $G$ and the available memory space that a semi-external algorithm can use. Assuming the nodes of $G$ and the node attributes of $\mathcal{A}$ can be represented by $32$-bit integers, the size of a simple graph $G$ is up to $64n(n-1)$ bits in that it may contain $n(n-1)$ edges. The available memory space of a semi-external algorithm based on the above assumption is normally no higher than $64\times 2n + 3\times 32n$ bits. That is because, besides its \textit{efficiency}, the \textit{minimum memory space requirement}~(MMSR) is also a main factor for measuring a semi-external algorithm~\cite{DBLP:journals/tkde/WanW23}, which is detailedly introduced below.

Normally, to ensure efficiency, semi-external graph algorithms follow one basic framework, as demonstrated in Figure~\ref{fig:basicFrameWork}.  \textit{They maintain a sketch $\mathcal{A}$ of $G$\footnote{The term ``sketch'' is defined in \cite{DBLP:journals/tkde/WanW23}. For the sake of simplicity, a sketch $\mathcal{A}$ of $G$ is normally a graph which can be generated by performing node-contraction operations, edge-contraction operations, or both on $G$.}, and solve the given graph problem by iteratively scanning the disk-resident input graph $G$. During each iteration, they use the scanned edge to enlarge $\mathcal{A}$. When $\mathcal{A}$ cannot be enlarged any further, they restructure $\mathcal{A}$ to a spanning tree $T$ of $G$ with their in-memory processes~(IMPs).} The MMSR of a semi-external algorithm refers to the size of $\mathcal{A}$. When $\mathcal{A}$ is a spanning tree of $G$ and there is no attribute maintained for each node in $\mathcal{A}$, the MMSR of a semi-external algorithm is the smallest. However, in this case, the algorithm has to invoke its IMP after each edge of $G$ is scanned and record all the intermediate results on disk, which is absolutely inefficient~\cite{DBLP:conf/spaa/SibeynAM02}. 

For efficiency, semi-external algorithms claim that $\mathcal{A}$ is a subgraph of $G$, and maintain certain attributes for each node in $\mathcal{A}$, so that they can process the scanned edges in batches to reduce the invocation times of their IMPs and record a portion of the intermediate results in memory. Since the smaller the MMSR of a semi-external algorithm, the better scalability it offers, existing semi-external algorithms maintain $2\times n$ edges in $\mathcal{A}$ and $2$ or $3$ attributes for each node in $\mathcal{A}$.

\begin{figure}[t]
	\centering
	\includegraphics[scale = 0.475]{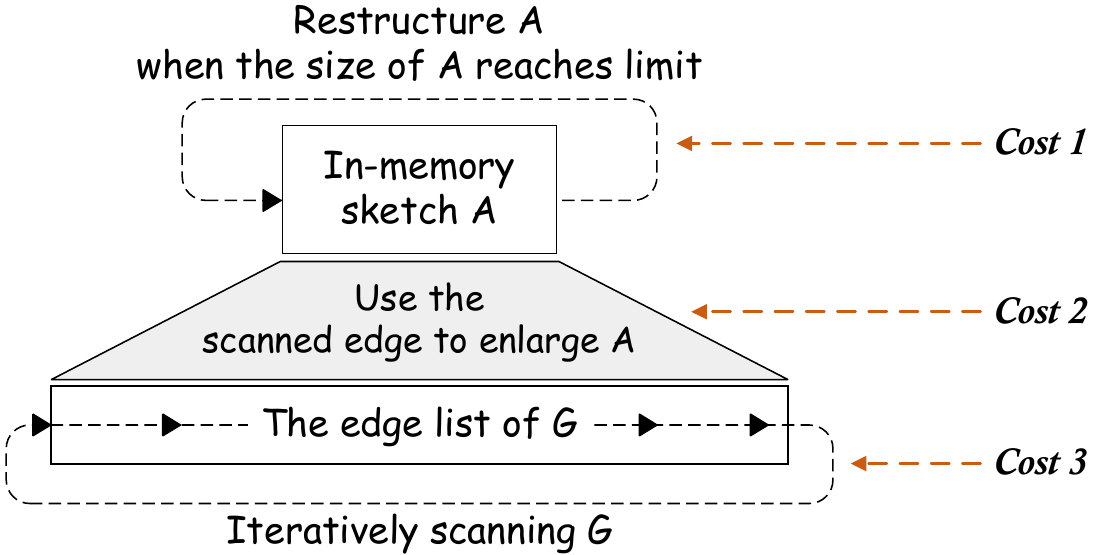}
	\caption{A schematic view of the basic framework of existing semi-external graph algorithms.}
	\label{fig:basicFrameWork}
\end{figure} 

\textbf{Contribution.} Motivated by the importance of addressing graph problems in semi-external memory model, this paper presents a comprehensive study of the semi-external BFS problem on large-scale directed graphs, which is difficult as discussed above.

This paper first introduces EE-BFS and EB-BFS, two naive semi-external BFS algorithms with a small MMSR, both of which  require storing two attributes for each node. As a natural way to address graph problems in semi-external memory model, both EE-BFS and EB-BFS  follow the basic framework demonstrated in Figure~\ref{fig:basicFrameWork}. 

EE-BFS maintains a spanning tree $T$ of $G$ as its in-memory sketch and restructures it during the edge scanning process. EE-BFS consumes $O(n\times m\times $\textit{LLSP}$(G))$ time, since it may restructure $T$ with all the scanned edges and scan $G$ \textit{LLSP}$(G)$ times, as proved in this paper. Here, \textit{LLSP}$(G)$ refers to the length of the longest simple path in $G$. 

EB-BFS maintains a subgraph of $G$ that is composed of $T$ and a set of edges $\mathbb{E}$ of $G$, as its in-memory sketch, and restructures it after scanning a batch of edges. It executes its IMP at most $\lceil\frac{m}{n}\rceil$ times during each iteration. Even though, with EB-BFS, the time consumption of obtaining BFS results in semi-external memory model is reduced to $O(m\times $\textit{LLSP}$(G))$, as proved in this paper, it is still expensive for practical applications that need to process large-scale graphs such as the graph WDC-2014 shown in Table~\ref{tab:real_datasets}.

With the high efficiency requirement of processing large-scale graphs in semi-external memory model, e.g. the graphs containing billions of nodes or billions of edges as shown in Table~\ref{tab:real_datasets}, a novel semi-external BFS algorithm, named EP-BFS, is proposed in this paper. EP-BFS is also designed with a small MMSR, i.e. in EP-BFS, each node in the in-memory sketch also has two attributes.

The efficiency of EP-BFS comes from the reduction of the three main costs in the semi-external graph algorithms based on the framework shown in Figure~\ref{fig:basicFrameWork}. The ability of EP-BFS to reduce the main costs mainly comes from the following aspects of EP-BFS. Firstly, the in-memory sketch $\mathcal{A}$ of EP-BFS only contains $\mathbb{E}$ and a portion of the spanning tree $T$ of $G$, denoted by $\mathbb{T}$. Secondly, EP-BFS introduces a threshold to control whether the newly scanned edge can be used to enlarge $\mathcal{A}$ or not. Besides, EP-BFS is a cache-friendly algorithm, since a novel in-memory sketch management strategy is designed in EP-BFS. Moreover, EP-BFS can further reduce $G$ and $\mathbb{T}$ at the end of each iteration. For more details, please see Section~\ref{sec:epBFs}.

Experiments are conducted on both real and synthetic datasets to evaluate the performance of the proposed algorithms. The real datasets refer to $14$ real networks, including social graphs, crawls from the Internet on different domains, etc., in which graph WDC-2014 has over $1.7$ billion nodes and graph eu-2015 has more than $91$ billion edges, as shown in Table~\ref{tab:real_datasets}. The generated synthetic graphs have up to $2$ billion nodes and $40$ billion edges, as discussed in Section~\ref{sec:experiments}. The experimental results on these graphs show that EP-BFS can efficiently process large-scale graphs in semi-external memory model, and its performance significantly surpasses EE-BFS and EB-BFS.

The contributions of this paper are listed as follows:
\begin{itemize}	
	\item[-] This paper presents a comprehensive study of the non-trivial semi-external BFS problem.
	
	\item[-] This paper presents two naive semi-external BFS algorithms with small MMSRs, namely EE-BFS and EB-BFS, based on the basic framework for addressing graph problems in semi-external memory model.
	
	\item[-] An efficient semi-external algorithm with a small MMSR, named EP-BFS, is introduced in this paper, which significantly reduces the main cost of computing BFS results in semi-external memory model.
	
	\item[-] This paper conducts extensive experiments on both real and synthetic large-scale graphs. Based on the experimental results, the high I/O and CPU efficiency of EP-BFS is confirmed.
\end{itemize}

The rest of this paper first surveys the related works in Section~\ref{sec:related_work}. Section~\ref{sec:preliminaries} gives the preliminaries. Two naive algorithms are given in  Section~\ref{sec:limitations}. EP-BFS is presented in Section~\ref{sec:epBFs}. Experimental evaluation is presented in Section~\ref{sec:experiments}. Section~\ref{sec:conclusion}
concludes the paper.

\section{Related Works} \label{sec:related_work}
As a basic graph traversal problem, many algorithms have been proposed to address the BFS problem on general directed graphs $G$, which are briefly summarized below.

The standard BFS algorithm is proposed on the internal memory model, where the size of the main memory is enough to hold the entire $G$~\cite{DBLP:books/daglib/0023376}. Specifically, the in-memory BFS algorithm traverses $G$ with an empty first-in-first-out queue $Q$. If BFS starts from the node $s$ of $G$, then it traverses $s$ by marking $s$ as \textit{visited} and adding all the adjacent nodes of $s$ to $Q$. Then, it runs iteratively. In each iteration, if $Q$ is not empty, it removes the first node $u$ from $Q$ and traverses $u$ if $u$ is unmarked; if $Q$ is empty, it randomly selects an unmarked node $v$ of $G$ and traverses $v$. Obviously, if $G$ has $n$ nodes and $m$ edges, the time and space costs of the in-memory BFS algorithm are $O(n+m)$, since it only traverses each node and edge in $G$ once, and $Q$ has $O(n+m)$ elements. However, the size of the main memory is hard to be sufficient to maintain current large-scale graphs entirely. 

External algorithms assume that the main memory has the capacity to reside $M$ elements, and each disk block can maintain $B$ elements~\cite{DBLP:conf/soda/BuchsbaumGVW00}. Based on this assumption, lots of external BFS algorithms are developed~\cite{DBLP:conf/soda/Meyer01,DBLP:conf/swat/ArgeT04,DBLP:conf/esa/MehlhornM02}. Even though some of them are efficient, on general directed graphs, the fast external BFS algorithm, named EM-BFS, still requires $O((n+\frac{m}{B})\log_2 \frac{n}{B} + sort(m))$ I/Os, where $sort(m)=O(\frac{m}{B}\log_{\frac{M}{B}}{\frac{n}{B}})$~\cite{DBLP:conf/soda/BuchsbaumGVW00}. EM-BFS is proposed by Buchsbaum et al. in \cite{DBLP:conf/soda/BuchsbaumGVW00}, based on the novel data structure named \textit{buffered repository tree} instead of the tournament tree utilized in previous works. A buffered repository tree $T_{br}$ requires $O(\frac{1}{B}\log_2\frac{N}{B})$ I/Os for its each insert operation, where $N$ represents the total number of insert operations that EM-BFS invokes during its one execution. $T_{br}$ requires $O(\log_2\frac{N}{B}+\frac{S}{B})$ I/Os for its each extract operation, where $S$ denotes the size of the result set of this extract operation. The basic idea of EM-BFS is to use $T_{br}$ to simulate the in-memory BFS algorithm, during which $m$ insert operations and $2\times n$ extract operations are executed on $T_{br}$.

Semi-external algorithms suppose that the size of the main memory has a lower bound, which cannot be arbitrarily small~\cite{DBLP:conf/spaa/SibeynAM02}. It is a key variant of external memory model, and is known as a good trade-off between the internal memory model~(efficiency) and the external memory model~(memory space usage). \textit{Processing graph problems under semi-external memory model has inevitability, since the commonly scale-out techniques will eventually hit the technology and monetary wall}~\cite{DBLP:conf/cidr/KerstenS17}. Generally, given a directed graph $G$, semi-external algorithms assume that the main memory can accommodate $c\times n$ elements of $G$ and $c$ is a small constant. Even though compared with external algorithms, semi-external algorithms can at least maintain a spanning tree $T$ of $G$, efficiently processing graph problems in semi-external memory model is still non-trivial~\cite{DBLP:journals/vldb/ZhangYQCL15}, as discussed in Section~\ref{sec:introduction}. Designing an efficient semi-external BFS algorithm on general directed graphs is still an open issue.

Thus, we are motivated to present a comprehensive study of semi-external BFS algorithms, and present an efficient semi-external BFS algorithm.

\textbf{Additionally.} Instead of the external computing models, such as semi-external model, the single-machine out-of-core graph processing systems~\cite{DBLP:conf/asplos/ZhangWZQHC18,DBLP:conf/usenix/AiZWQCZ17,DBLP:conf/usenix/Vora19,DBLP:conf/usenix/ZhuHC15} are also proposed to process large-scale graphs with memory constraints. Even though such graph processing systems can efficiently report the results of the BFS problem which starts from a specific node $u$ and ends until all the nodes in $G$ reached by $u$ are visited, they are inefficient for computing \textit{the total BFS order for a relatively large graph $G$}~(the type of BFS problem studied in this paper), as we proved in Section~\ref{sec:experiments}. The reason is obvious, that there are often numerous sparse subgraphs or even numerous isolated nodes in real networks~\cite{LKAQ}, for example the graphs shown in Table~\ref{tab:real_datasets}, which leads to a large number of random accesses to memory or disk, i.e. high computation cost.

\section{Problem Statement} \label{sec:preliminaries}
We focus on the semi-external BFS problem on directed disk-resident graphs, as defined below. Table~\ref{tab:notation} shows all the frequently-used notations of this paper.

\begin{table}[t]
	\centering \caption{Frequently-used notations.} \label{tab:notation}
	\small
	\begin{tabular}{c|p{10cm}}
		\hline
		\textbf{Notation} &  \textbf{Description}\\%
		\hline
		$G$ &The input graph with $n$ nodes and $m$ edges\\
		
		$r$&A dummy node of $G$ which connects to all the other nodes and BFS runs from it by default\\
		$bfo(v,T)$	&	The breadth-first order of $v$ on tree $T$\\
		
		
		$T$ & An in-memory spanning tree of $G$ rooted at $r$, which is composed of  $\mathbb{T}$ and $E_\mathbb{T}$ in EP-BFS\\
		
		$\mathbb{E}$ & An edge list containing a portion of edges of $G$\\
		
		$\mathcal{A}$ & The in-memory sketch of $G$\\
		
		\textit{LLSP}$(G)$& The length of the longest simple path in $G$\\
		
		$K$& A concrete number, where $Kn$ is the maximum number of edges that $\mathbb{E}$ can reside with the memory constraint\\
		
		\hline
	\end{tabular}
\end{table}

\begin{definition}
	\label{def:DG}
	\textbf{(Directed graph).} A directed graph $G$ is a tuple $(V,E)$, where (\romannumeral1) $V$ and $E$ are the node set and edge set of $G$, respectively; (\romannumeral2) $E\subseteq V\times V$, and, $\forall e=(u,v)\in E$, $e$ is a directed edge from node $u$ to node $v$.
\end{definition}

For simplicity, given a graph $G$, we let $V(G)$ and $E(G)$ be the node set and edge set of $G$, respectively, and assume that $|V(G)|=n$ and $|E(G)|=m$. Notation $N^+(u,G)=\{v|\forall (u,v)\in E(G)\}$ represents the \textit{out-neighborhood} of $u$ on $G$. As an example, given a directed graph $G_0$ shown in Figure~\ref{fig:introduction_figure}, $V(G_0)=\{a, b, c,\dots\}$, $E(G_0)=\{(a, h),\dots\}$, $N^+(a,G_0)=\{h,f,s,u\}$, $N^+(b,G_0)=\{d,a\}$, etc. \textit{LLSP}$(G)$ is used to denote the length of the longest simple path of $G$. Here, a directed path $p$ of $G$ from node $v_1$ to node $v_k$ can be represented by $(v_1,v_2,v_3,\dots, v_k)$, and $\forall i\in [1,k)$, $(v_i,v_{i+1})\in E(G)$. $p$ is a simple path of $G$ iff $\forall i,j\in [1,k]$ if $i\neq j$ then $v_i\neq v_j$. 

Given a directed graph $G$, \textit{\textbf{breadth-first search~(BFS)}} traverses $G$ with an empty queue $Q$, and runs until it visits all the nodes in $G$. When BFS visits a node $u$ of $G$, it first marks $u$ as \textit{visited}, and then puts all the nodes in $N^+(u,G)$ into $Q$. BFS visits a node $u$ of $G$ iff, (\underline{\textit{Case 1}}) $u$ is the first node in $Q$ and $u$ has not been marked as \textit{visited} yet, or (\underline{\textit{Case~2}}) $Q$ is empty and $u$ has not been marked. The total order that BFS visits the nodes in $G$ is called \textit{breadth-first order}, which is not unique. For example, the breadth-first order of $G_0$ in Figure~\ref{fig:introduction_figure}(a) could be (\romannumeral1) $a,h,f,s,u,b,c,d,g,p,q,t$, (\romannumeral2) $t,q,g,d,c,b,a,h,f,s,u,p$, (\romannumeral3) $t,q,g,p,d,c,h,b,a,f,s,u$, etc. 

\begin{figure}[t]
	\centering
	
	\renewcommand{\thesubfigure}{}
	\subfigure[(a) $T_0$]{
		\includegraphics[scale = 0.47]{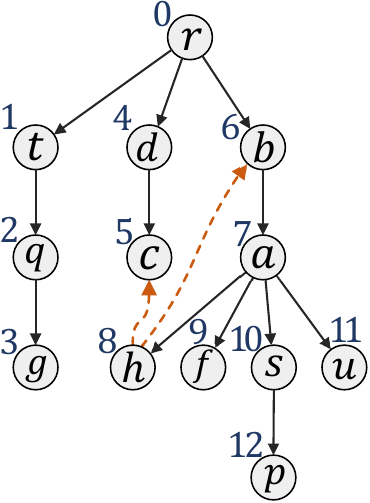}}\quad\,
	\subfigure[(b) $T_1$]{
		\includegraphics[scale = 0.47]{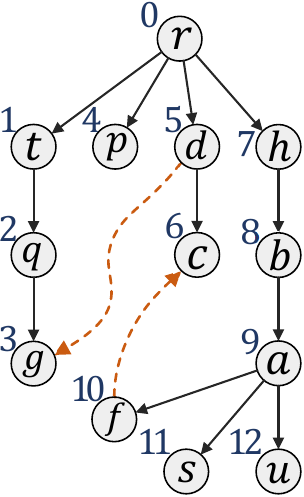}}\quad\,
	\subfigure[(c) $T_2$]{
		\includegraphics[scale = 0.47]{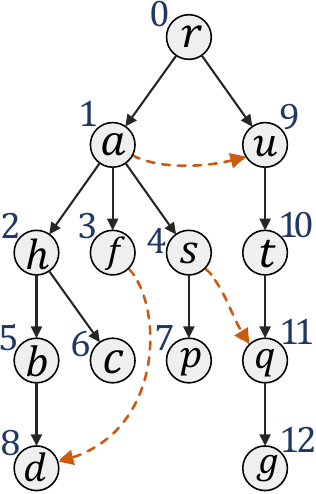}}\quad\,
	\subfigure[(d) $T_3$]{
		\includegraphics[scale = 0.47]{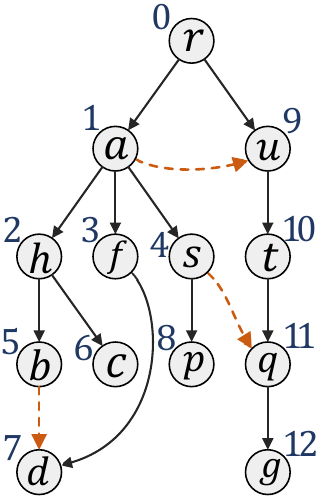}}\quad\,
	
	\caption{A schematic view of four spanning trees $T_0,T_1,T_2$ and $T_3$ of $G_0$. Note that, the edges of $G_0$ include both the solid lines and the dotted lines, whereas only the solid lines represent the tree edges of $T_0,T_1,T_2$ and $T_3$.}
	\label{fig:trees_T0_T3}
\end{figure} 

Notation $bfo(u,T)$ is used to denote the breadth-first order of node $u$ on $T$, and we assume $G$ contains a dummy node $r$. Each spanning tree $T$ of $G$ is an \textit{ordered tree} rooted at $r$, where (\romannumeral1) if $w$ is a node of $T$, and nodes $u$ and $v$ are children of $w$ in $T$, then $u$ is the \textit{left brother} of $v$ in $T$ iff $bfo(u,T)<bfo(v,T)$; (\romannumeral2) a node $x$ is visited in the above \underline{\textit{Case~2}} iff $x$ is a child of $r$. As an example, Figure~\ref{fig:trees_T0_T3} demonstrates three spanning trees $T_0,T_1,T_2$ and $T_3$, in which the breadth-first order of each node on the tree is also given. For instance, $bfo(t,T_0)=1$, $bfo(q,T_0)=2$, $\dots$, $bfo(p,T_0) = 12$.


As BFS visits each node in $G$ only once, each breadth-first order of $G$ corresponds to a \textit{BFS-tree} $T'$ of $G$, where (\romannumeral1) $T'$ is a spanning tree, (\romannumeral2) the children of $r$ in $T'$ are visited when $Q$ is empty, and (\romannumeral3) $\forall u\in V(T')$, if $v$ is a child of $u$ in $T'$ then $v$ is added to $Q$ after $u$ has been visited. That is, a spanning tree $T$ of $G$ is a BFS-tree of $G$ iff  classifying $G$ with $T$, there is no such edge $e'=(u,v)$ in $G$, where (\romannumeral1) $bfo(u,T)<bfo(v,T)$, and (\romannumeral2) supposing the parent of $v$ on $T$ is $w$, $w\neq r$ and $bfo(w,T)\leq bfo(u,T)$. For the sake of simplicity, we name the edges like $e'$ as \textit{V-BFS edges}. Note that, as far as we know, this concept of V-BFS edge is first proposed in this paper. For instance, in Figure~\ref{fig:trees_T0_T3}(c), edge $(a,u)$ is a V-BFS edge of $G_0$, as classified by $T_2$. 

Based on the above discussions, Definition~\ref{def:bfs-tree} defines BFS-tree.

\begin{definition}
	\label{def:bfs-tree}
	\textbf{(BFS-tree).} A spanning tree $T$ of $G$ is a BFS-tree of $G$ iff there is no V-BFS edge in $G$ as classified by $T$.
\end{definition}

For instance, $T_0,T_1$ and $T_2$~(shown in Figure~\ref{fig:trees_T0_T3}) are spanning trees of $G_0$~(shown in Figure~\ref{fig:introduction_figure}(a)). $T_0$ and $T_1$ are BFS-trees of $G_0$, while $T_2$ is not. That is because there are no V-BFS edges in $G_0$ as classified by $T_0$ or $T_1$, but $(a,u)$ is a V-BFS edge of $G_0$ as classified by $T_2$.

\textbf{Problem statement.} This paper focuses on the problem of computing a BFS-tree of $G$ in semi-external memory model efficiently, where $G$ denotes a disk-resident directed graph. Semi-external memory model assumes that the main memory can accommodate a small portion of $G$, but cannot maintain the entire $G$. 

\section{Naive algorithms}\label{sec:limitations}

Processing graph problems on semi-external model is intricate, since algorithms can only reside on a small proportion of $G$ in the main memory. To efficiently address the complicated semi-external BFS problem, we first propose two naive semi-external BFS algorithms, namely EE-BFS~(Algorithm~\ref{algo:ee_bfs}) and EB-BFS~(Algorithm~\ref{algo:eb_bfs}), in this section.

EE-BFS and EB-BFS follow the basic framework of traditional semi-external algorithms, as depicted in Figure~\ref{fig:basicFrameWork}. They maintain a sketch $\mathcal{A}$ of $G$ in the main memory, and gradually restructure $\mathcal{A}$ until the result can be computed based on $\mathcal{A}$. The difference between EE-BFS and EB-BFS is that EE-BFS processes $G$ edge by edge, while EB-BFS processes $G$ in edge batches, which is discussed in detail below. 

In EE-BFS, $\mathcal{A}$ is a spanning tree $T$ of $G$, where initially~(Line~\ref{line:eeBFS:initialize_T}, Algorithm~\ref{algo:ee_bfs}) the root of $T$ is the dummy node $r$ of $G$, and all the other nodes in $G$ are the children of $r$ in $T$. Then, EE-BFS runs with iterations~(Lines~\ref{line:eeBFS:while_start}-\ref{line:eeBFS:while_end}), and in each iteration, it scans $G$ once~(Lines~\ref{line:eeBFS:for_start}-\ref{line:eeBFS:for_end}). For each scanned edge $e=(u,v)$, EE-BFS restructures $T$ in Line~\ref{line:eeBFS:restructure} when $e$ is a V-BFS edge. EE-BFS is terminated in the $k$th iteration, if each scanned edge in this iteration is not a V-BFS edge, i.e. $T$ is a BFS-tree of $G$ based on Definition~\ref{def:bfs-tree}.

\begin{figure}[t]
	\centering
	
	\renewcommand{\thesubfigure}{}
	
	\subfigure[(a) $T_4$]{
		\includegraphics[scale = 0.5]{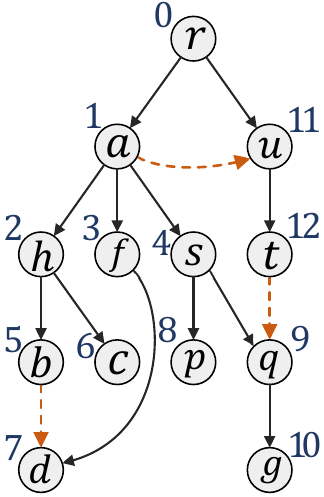}}\quad\,
	\subfigure[(b) $T_5$]{
		\includegraphics[scale = 0.5]{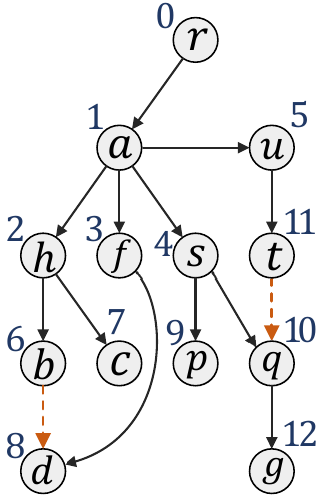}}\quad\,
	\caption{The schematic view of the spanning trees $T_4$ and $T_5$ of $G_0$, where the black solid lines are their tree edges.}
	\label{fig:trees_T4T5}
\end{figure} 

For instance, assuming the edges of $G_0$~(Figure~\ref{fig:trees_T0_T3} and Figure~\ref{fig:trees_T4T5}) are in the following order $\dots, (f,d), (s,q), (a,u), \dots$ When EE-BFS scans edge $(f,d)$ in Line~\ref{line:eeBFS:for_start}, if $T$ is in the form of $T_2$~(Figure~\ref{fig:trees_T0_T3}(c)), then the restructured $T$ at the end of this iteration is in the form of $T_3$~(Figure~\ref{fig:trees_T0_T3}(d)). After that, EE-BFS restructures $T_3$ into $T_4$~(Figure~\ref{fig:trees_T4T5}(a)) since it scans edge $(s,q)$, and restructures $T_4$ into $T_5$~(Figure~\ref{fig:trees_T4T5}(b)) because of edge $(a,u)$.

In EB-BFS, $\mathcal{A}$ is a subgraph of $G$, and is composed of $T$ and $\mathbb{E}$. EB-BFS also processes $G$ with iterations, and in each iteration, it scans $G$ once, as shown in Lines~\ref{line:ebBFS:while_start}-\ref{line:ebBFS:while_end}. The initialization of $T$ in EB-BFS and that in EE-BFS are the same. However, different from EE-BFS, EB-BFS only restructures $T$ when it scans $Kn$ edges of $G$, as demonstrated in Lines~\ref{line:ebBFS:for_start}-\ref{line:ebBFS:for_end}. Here, $K$ refers to a concrete number, by default, $K=1$. The $Kn$ edges are maintained in the edge list $\mathbb{E}$ in the order that they are scanned. EB-BFS computes the BFS-tree $T'$ of the graph composed of $T$ and $\mathbb{E}$ based on Stipulation~\ref{sti:stipulation}. 

\begin{stipulation}
	\label{sti:stipulation}
	Assuming $\mathcal{G}$ is the graph composed of $T$ and $\mathbb{E}$. For each node $u$ in $\mathcal{G}$, after traversing $u$, (\romannumeral1) BFS first adds the elements in $N^+(u,T)$ to its search queue from the leftmost child of $u$ on $T$ to the rightmost child of $u$ on $T$, and then  (\romannumeral2) adds $v$ into $T$ if $\exists (u,v)\in \mathbb{E}$ from front to back.
\end{stipulation}

EB-BFS can be terminated in its $k$th iteration, iff, in the $k$th iteration of the loop in Lines~\ref{line:ebBFS:while_start}-\ref{line:ebBFS:while_end}, each $T'$ has the same edge set as $T$~(Line~\ref{line:ebBFS:update_true}). For instance, if at the beginning of one iteration of the loop in Lines~\ref{line:ebBFS:for_start}-\ref{line:ebBFS:for_end}, Algorithm~\ref{algo:eb_bfs}, $T$ is in the form of $T_2$~(Figure~\ref{fig:trees_T0_T3}(c)), $\mathbb{E}=\{(f,d),(s,q),(a,u)\}$, then at the end of this iteration $T$ is in the form of $T_5$~(Figure~\ref{fig:trees_T4T5}(b)).

\begin{algorithm}[t]
	\caption{EE-BFS$(G)$}
	\label{algo:ee_bfs}
	\begin{algorithmic}[1]
		\renewcommand{\algorithmicrequire}{\textbf{Input:}}
		\renewcommand{\algorithmicensure}{\textbf{Output:}}
		\renewcommand{\algorithmiccomment}[1]{  #1}
		
		\REQUIRE $G$ is a graph stored in disk.
		\ENSURE  The BFS-tree of $G$.
		
		\STATE Initialize tree $T$, \textit{update}$\,\gets\,$\textit{true}\label{line:eeBFS:initialize_T}
		
		\STATE \textbf{while} \textit{update}$\,=\,$\textit{true} \textbf{do}\label{line:eeBFS:while_start}
		\STATE \quad \textit{update}$\,\gets\,$\textit{false}\label{line:eeBFS:while_update}
		\STATE \quad \textbf{for} each $e=(u,v)\in E(G)$ \textbf{do} \label{line:eeBFS:for_start}
		\STATE \quad\quad \textbf{if} $e$ is a V-BFS edge \textbf{then}\label{line:eeBFS:forIf}
		\STATE \qquad\quad Letting $v$ be the rightmost child of $u$ on $T$\label{line:eeBFS:restructure}
		\STATE \qquad\quad  \textit{update}$\,\gets\,$\textit{true}			\label{line:eeBFS:update_change}
		\STATE \quad \textbf{end for}\label{line:eeBFS:for_end}
		\STATE \textbf{end while}\label{line:eeBFS:while_end}
		
	\end{algorithmic}
\end{algorithm}

\begin{algorithm}[t]
	\caption{EB-BFS$(G)$}
	\label{algo:eb_bfs}
	\begin{algorithmic}[1]
		\renewcommand{\algorithmicrequire}{\textbf{Input:}}
		\renewcommand{\algorithmicensure}{\textbf{Output:}}
		\renewcommand{\algorithmiccomment}[1]{  #1}
		
		\REQUIRE $G$ is a graph stored in disk.
		\ENSURE  The BFS-tree of $G$.
		
		\STATE Initialize tree $T$, \textit{update}$\,\gets\,$\textit{true}\label{line:ebBFS:initialize_T}
		
		\STATE \textbf{while} \textit{update}$\,=\,$\textit{true} \textbf{do}\label{line:al:while}\label{line:ebBFS:while_start}
		\STATE \quad \textit{update}$\,\gets\,$\textit{false}\label{line:ebBFS:while_update} 
		\STATE \quad \textbf{for} each $Kn$ edges $\mathbb{E}$ of $G$ \textbf{do} \hfill \COMMENT{// $|\mathbb{E}|=Kn$} \label{line:al:for}\label{line:ebBFS:for_start}
		\STATE \quad \quad $T'\gets\,$IM-BFS$(T,\mathbb{E})$\label{line:al:IM_BFS}
		\STATE \quad  \quad \textit{update}$\,\gets\,$\textit{true if $E(T)\neq E(T')$}\label{line:ebBFS:update_true}	
		
		\STATE \quad \quad $T\gets T'$ \COMMENT{// The BFS-tree of graph $\big(V(T),E(T)\cup\mathbb{E}\big)$}\label{line:al:reduction}\label{line:ebBFS:reSet_T}
		\STATE \quad \textbf{end for}\label{line:ebBFS:for_end}
		\STATE \textbf{end while}\label{line:ebBFS:while_end}
		
	\end{algorithmic}
\end{algorithm}

\textbf{Analysis.} In this part, we discuss the correctness and complexities of EE-BFS and that of EB-BFS. Firstly, Theorem~\ref{theorem:eb_searchTime} states that EE-BFS and EB-BFS can be terminated and scan $G$ at most \textit{LLSP}$(G)$ times. 


\begin{theorem}
	\label{theorem:ee_searchTime}\label{theorem:eb_searchTime}
	EE-BFS and EB-BFS can be terminated and scans $G$ at most \textit{LLSP}$(G)$ times.
\end{theorem}
\begin{proof}	
	Below, we use $depth(u,T)$ to represent the depth of a node on a spanning tree $T$. Supposing that root $r$ reaches $u$ by passing through at least $k$ edges of $T$, i.e. $T$ has a path $(r~(v_0), v_1,\dots, u~(v_k))$ of length $k$, then \textit{depth}$(u,T)=k$.
	
	\textit{\textbf{(\romannumeral1)}} Assuming the leftmost child of $r$ can reach all the other nodes of $G$. This statement is valid iff, \textit{at the end of the $i$th iteration,  $\forall u\in V(G)$, if \textit{depth}$(u,T)\leq i$, then $bfo(u,T)$ is fixed.} 	We prove the following statement by induction on $i$. Firstly, this statement is certainly correct when $i=1$, since if a node $x$ on $T$ and $depth(x,T)\leq 1$ then $x=r$ or $x$ is the leftmost child of $r$. Assuming the following statement is correct when $i=k$, then the following statement is certainly correct:``at the end of the $k$th iteration of the loop in Lines~\ref{line:eeBFS:while_start}-\ref{line:eeBFS:while_end}, Algorithm~\ref{algo:ee_bfs},  $\forall u\in V(G)$, if \textit{depth}$(u,T)\leq k$, then $bfo(u,T)$ is fixed.'', since both EE-BFS and EB-FBS scan the entire $G$ during the $(k+1)$th iteration. 
	
	\textit{\textbf{(\romannumeral2)}} Assuming the leftmost child of $r$ cannot reach all the other nodes of $G$. Supposing in the $i$th iteration $r$ has $\gamma$ out-neighbors on $T$, and $T_\alpha$ represents the cut-tree of $T$ rooted at the $\alpha$th child of $r$. Then, it could be easily proved based on (\romannumeral1) that at the end of the $i$th iteration,  $\forall u\in V(G)$, if \textit{depth}$(u,T)\leq i$ and $u\in V(T_\alpha)$, then either (1) the position of $u$ on $T$ is fixed, or (2) $u$ belongs to $T_\beta$~($\beta<\alpha$) and the $\beta$th child of $r$ on $T$ needs to pass more than $depth(u,T)$ nodes for reaching $u$ when $T$ is a BFS-tree of $G$. 
\end{proof}

The time cost of EE-BFS is $O\big(n\times m\times\,$\textit{LLSP}$(G)\big)$, in that EE-BFS needs to restructure $T$ with each scanned edge, it has to update the the breadth-first order of the nodes on $T$, and the time cost of updating the breadth-first order is $O(n)$. The time cost of EB-BFS is $O\big(\lceil\frac{m}{n}\rceil\times n\times\,$\textit{LLSP}$(G)\big)=O\big(m\times\,$\textit{LLSP}$(G)\big)$, since EB-BFS needs to execute its IMP $\lceil\frac{m}{n}\rceil$ times in one iteration, each execution of which consumes $O(n)$ time. \textit{The I/O costs of EE-BFS and EB-BFS are $O\big(\frac{m}{B}\times $\textit{LLSP}$(G)\big)$}, since they scan the disk-resident $G$ once in each iteration. $B$ is the size of one disk block.

The MMSR of EE-BFS and that of EB-BFS are $2\times n$, since they need to maintain two attributes for each node $v$ on $T$, i.e. (\romannumeral1) the breadth-first order of $v$ on $T$ and (\romannumeral2) the parent of $v$ on $T$. Besides, to the best of our knowledge, there is no technique that is proposed for packing these two attributes, because in semi-external algorithms the values of their node attributes are changed constantly and cannot be estimated in advance. For efficiency, current semi-external algorithms maintain their node attributes in a contiguous memory space.

The correctness of EE-BFS and that of EB-BFS are proved in Theorem~\ref{theorem:correctness_eeBfs}, respectively.
\begin{theorem}
	\label{theorem:correctness_eeBfs}
	EE-BFS and EB-BFS return $T$ correctly.
\end{theorem}
\begin{proof}
	This statement is valid, because \textit{in the $j$th~($j>i$) iteration,  $\forall e$$=$$(u,v)\in E(G)$, if  $depth'(u,T)$$\leq$$i\wedge depth'(v,$ $T)$$\leq$$i$, $e$ is not a V-BFS edge.} Here, notation $depth'(v,T)$ is different from $depth(v,T)$. Letting $T_i$ and $T_j$ represent the cuttrees of $T$ rooted at the $i$th and the $j$th leftmost children of $r$, respectively;  $d(T_i)=\max\{depth(x,T), \forall x\in V(T_i)\}$ represents the maximum depth of the nodes on $T_i$; $v$ belongs to $T_j$. $depth'(v,T)=\Sigma_{1\leq i\leq j-1} d(T_i)+depth(v,T_j)$. 
\end{proof}

\section{EP-BFS: The efficient strategy}\label{sec:epBFs}

\textbf{Overview.} EP-BFS maintains two attributes for each node $u$ in $G$, i.e. $u.\mathcal{B}$ and $u.\mathcal{P}$, that is the MMSR of EP-BFS is also limited to $2\times n$. Despite this, EP-BFS is still CPU and I/O efficient, because the main costs of semi-external graph algorithms based on the basic framework shown in Figure~\ref{fig:basicFrameWork} are reduced in EP-BFS. Here, the main costs include:
\begin{itemize}
	\item[-] \textit{(Cost~1:)} \textit{the total cost of executing the IMP of a semi-external algorithm to restructure its in-memory sketch $\mathcal{A}$}, which is related to \textit{(c-\romannumeral1)} the total invocation times of the IMP and \textit{(c-\romannumeral2)} the cost of executing the IMP once; 
	
	\item[-] \textit{(Cost~2:)} \textit{the total cost of using the scanned edges to enlarge $\mathcal{A}$}, which is dependent on \textit{(c-\romannumeral3)} the total number of scanned edges added into $\mathcal{A}$, and \textit{(c-\romannumeral4)} the time cost of enlarging $\mathcal{A}$ with a scanned edge;
	
	\item[-] \textit{(Cost~3:)} \textit{the total cost of scanning the edge list of $G$}, which corresponds to \textit{(c-\romannumeral5)} the total number of I/Os for scanning the disk-resident edges.
\end{itemize}

Below, a brief introduction of EP-BFS is presented, in which the reason why EP-BFS can reduce the above three main costs is discussed in detail.

Firstly, the in-memory sketch of EP-BFS only contains $\mathbb{E}$ and \textit{a portion of a spanning tree $T$}~(denoted by $\mathbb{T}$) for $G$ rather than $\mathbb{E}$ and the complete $T$. That is, EP-BFS decomposes $T$ into two parts, where one is an in-memory subgraph  $\mathbb{T}$ of $G$, and another is a disk-resident edge set $E_T$. Instead of processing $G$ directly, EP-BFS first scans the input graph $G$. In the meantime, EP-BFS can identify the positions of certain nodes on $T$ when $T$ is a BFS-tree of $G$, and then removes all such nodes from $G$. 

Note that the reduction of the node set maintained in the main memory greatly affects the efficiency of a semi-external algorithm. That is because if a node $u$ is maintained in the main memory, at least an edge related to $u$ should be residing in the main memory. Assuming $G_R$ denotes the reduced $G$, and the computing device can accommodate $(1+K)n$ edges in the main memory, then the upper bound of the times that the IMP of EP-BFS is executed during one iteration is lowered from $\lceil\frac{m}{Kn}\rceil$ to $\lceil\frac{|E(G_R)|}{Kn+n-|V(G_R)|}\rceil$. That is, \textit{c-\romannumeral1} is reduced.

The fewer times the IMP of EP-BFS is executed, the more significant the decrease in computational cost of EP-BFS, as  the IMP of a semi-external has a high number of invocations, especially for large-scale graphs. Even though the in-memory computation is fast, the time cost of executing the IMP of a semi-external algorithm cannot be neglected~\cite{10.5555/1841497}. For example, graph gsh-2015 utilized in our experiments contains $988$ million nodes, as shown in Table~\ref{tab:real_datasets}.

\begin{figure}[t]
	\centering
	\includegraphics[scale = 0.7]{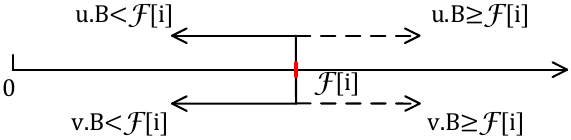}
	\caption{The relationships of an edge $e=(u,v)$ and the threshold $\mathcal{F}[i]$.}
	\label{fig:IMin}
\end{figure} 

More importantly, EP-BFS utilizes a threshold $\mathcal{F}[i]$ to reduce both \textit{c-\romannumeral1} and \textit{c-\romannumeral3}. Here, the notation $i$ in $\mathcal{F}[i]$ represents the number of times that the IMP of EP-BFS has been executed in the current iteration. The determination of whether a newly scanned edge is added to the in-memory sketch or not is based on $\mathcal{F}[i]$. $\mathcal{F}[i]$ is initialized to $+\infty$. The value of $\mathcal{F}[i]$ is updated to a smaller value $k$ only if the newly scanned edge $(u,v)$ is a V-BFS edge as classified by $T$~(the combination of $\mathbb{T}$ and $E_\mathbb{T}$), and $u.\mathcal{B}=k$. In EP-BFS, the edges $e=(u,v)$ of $G$ are classified based on the relationship between the attributes $\mathcal{B}$ and $\mathcal{F}[i]$, as shown in Figure~\ref{fig:IMin}. $e$ is utilized to enlarge the in-memory sketch, when both $u.\mathcal{B}$ and $v.\mathcal{B}$ are equal to or larger than $\mathcal{F}[i]$, or $e$ is a V-BFS edge as classified by $T$.

In addition, EP-BFS fully utilizes cache to decrease \textit{c-\romannumeral2} and \textit{c-\romannumeral4}, in that a novel strategy is proposed for efficiently managing its in-memory sketch $\mathcal{A}$. Traditionally, as a semi-external algorithm has to constantly insert edges into $\mathcal{A}$ or remove edges from $\mathcal{A}$, it is a natural way to use two arrays \textit{Adj}~(length of $n$) and \textit{ES}~(length of $2n\times K$) for residing $\mathcal{A}$ in the main memory, since it can protect the performance of a semi-external algorithm from being affected by frequently requesting and freeing up memory space or being affected by the garbage collection mechanisms of different implementation languages. Here, $K>1$ because semi-external algorithms need to reside at least $n$ edges in memory. 

For graphs with a small node set, only utilizing these two arrays without any other strategies is sufficient, in that when the elements in the arrays need to be frequently accessed, computing devices can load the entire arrays into cache. When processing graphs with numerous nodes, e.g. graph gsh-2015 in Table~\ref{tab:real_datasets} has over $988$ million nodes, the sizes of the arrays are too large to reside in cache~\cite{10.5555/1841497}. Randomly accessing the elements of the two arrays is inevitable when searching in $\mathcal{A}$ or enlarging $\mathcal{A}$, which leads to inefficiency.

To decrease the number of times that EP-BFS accesses the elements maintained in \textit{Adj} and \textit{ES} randomly, and to achieve high cache efficiency, EP-BFS stores the out-neighbors of a node $u$ of $\mathcal{A}$ in \textit{ES} contiguously, which is non-trivial. The difficulty of managing the edges in \textit{ES} with a specified order not only comes from \textit{the fact that the scale of $G$ can be large or the assumption that the edges of $G$ are resident on disk randomly}, but also because \textit{semi-external algorithms need to insert and delete the edges in memory frequently}.

In other words, this is a key contribution of EP-BFS, since (\romannumeral1) after storing edges in \textit{ES}, the random accesses to \textit{ES} caused by the operations of visiting the out-neighborhood of a node in $\mathcal{A}$ can be greatly reduced; (\romannumeral2) semi-external graph algorithms need to frequently and repeatedly visit the   out-neighborhood of a node in $\mathcal{A}$.

Finally, EP-BFS reduces \textit{c-\romannumeral5}, because it reorganizes the edge set $E_R$ of $G_R$ in the form of an adjacent list to achieve a smaller storage volume~(\textit{I/O efficiency}), which can also help EP-BFS to obtain a better cache utilization~(\textit{CPU efficiency}). EP-BFS scans $E_R$ with iterations rather than $E(G)$. During each iteration, it can further reduce $G$ and shrink $\mathbb{T}$, until $\mathbb{T}$ is an empty tree. Note that the smaller the size of $\mathbb{T}$, the more newly scanned edges can be stored in the main memory, and the fewer times that the IMP of EP-BFS is executed during one iteration.

\subsection{Algorithm description}\label{sec:epBFs:algorithmDescription}

In this part, we introduce  EP-BFS in detail, based on its pseudocode shown in Algorithm~\ref{algo:ep_bfs}. 

\begin{algorithm}[t]
	\caption{EP-BFS$(G)$}
	\label{algo:ep_bfs}
	\begin{algorithmic}[1]
		\renewcommand{\algorithmicrequire}{\textbf{Input:}}
		\renewcommand{\algorithmicensure}{\textbf{Output:}}
		\renewcommand{\algorithmiccomment}[1]{  #1}
		
		\REQUIRE $G$ is a graph stored in disk.
		\ENSURE  The BFS-tree of $G$.
		
		\STATE $E_1,E_2,\dots, E_k\gets\,$ ScanG(G)\label{line:pg:ScanG}, $E_R\gets[], \ E_o\gets[], \ E_i\gets[], \ \mathbb{E}\gets[],\ E_\mathbb{T}\gets[]$ \label{line:pg:EREOEIE}
		\STATE $\mathbb{T}\gets$ a forest with nodes $v_1,v_2,\dots,v_n$ and zero edges
		
		\STATE \textbf{for} each node $u\in V(G)$ \textbf{do}\label{line:pg:for:start}
		
		\STATE\quad $V_u,E_o\gets$Read($E_1\dots, E_k,u,E_o$)\label{line:pg:for:read}
		
		\STATE\quad \textit{If $id(u)$$=$$0,\ \forall v\in V_u$, add $(u,v)$ into $E_i$ and \textbf{goto} Line}~\ref{line:pg:for:start}\label{line:pg:for:if1}
		
		\STATE\quad $\forall v\in V_u$, add $(u,v)$ into $E_R$ and $\mathbb{E}$\label{line:pg:for:otherwise}
		\STATE\quad \textit{When $|\mathbb{E}|+|\mathbb{T}|$ reaches limit, execute }Reduce($\mathbb{T}$, $\mathbb{E}$)\label{line:pg:for:reduce}
		\STATE\textbf{end for}\label{line:pg:for:end}, \quad$\mathbb{T}\gets\,$TreeReduce($\mathbb{T}$, $\mathbb{E}$)\label{line:pg:Treduce} 
		
		\STATE $\mathbb{E}\gets[]$, $E_{next}\gets\phi$,  $\,\mathcal{F}_{R}\gets-\infty$, $\,\mathcal{F}_{C}\gets-\infty$, $\,f\gets 1$, $\,t\gets 0$\label{line:ep:intializedOthers}
		
		\STATE \textbf{while} $\mathcal{F}_{R}<n$  \textbf{do}\label{line:ep:while:start}
		\STATE\quad $i\gets 0$, $\mathcal{F}[i]\,\gets +\infty$\label{line:ep:while:initialize}
		\STATE \quad \textbf{for} each edge $e=(u,v)$ in $E_{R}$ \textbf{do} \label{line:ep:while:for:start}
		
		\STATE\qquad\textit{If $u.\mathcal{B}\negmedspace\leq\negmedspace \mathcal{F}_{R}$ \textbf{or} $v.\mathcal{B}\negmedspace\leq\negmedspace \mathcal{F}_{C}$ then  \textbf{goto} Line}~\ref{line:ep:while:for:start} \label{line:ep:while:for:if1_GOTOL22}
		
		\STATE\qquad \textit{If $E_{next}\neq \phi$ then} $E_{next}\gets\,$Enlarge$(E_{next}, e,\mathcal{F}_{CC})$\label{line:ep:while:for:if2_englargeE}
		
		
		\STATE\quad\quad \textbf{if} $\mathcal{F}[i]\,\leq u.\mathcal{B}$ \textbf{and} $v.\mathcal{P}\neq u$ \textbf{then}\label{line:ep:while:for:if3_insertFlage}
		
		\STATE\quad\qquad \textit{If $v.\mathcal{B}\geq\,$$\mathcal{F}[i]$ then }$\mathbb{E}\gets\mathbb{E}\cup\big\{(u,v)\big\}$\label{line:ep:while:for:if3:if}
		
		\STATE\quad\quad \textbf{else if} $u.\mathcal{B}<v.\mathcal{B}$ \textbf{and} $u.\mathcal{B}<v.\mathcal{P}.\mathcal{B}$ \textbf{then}\label{line:ep:while:for:if3:elseif}
		
		\STATE\quad\qquad $\mathcal{F}[i]\,\gets u.\mathcal{B}$, $\,\mathbb{E}\gets\mathbb{E}\cup\big\{(u,v)\big\}$\label{line:ep:while:for:if3:elseif_process}
		
		\STATE\quad\quad\textbf{end if}\label{line:ep:while:for:if3:endif}

		\STATE\qquad \textbf{if} $|\mathbb{E}|+|\mathbb{T}|=(K+1)n$ \textbf{then}\label{line:ep:while:for:if4}
		
		\STATE\qquad\quad$\mathbb{T}$$\gets$EP-Reduce$(\mathbb{T},\mathbb{E})$, $\mathbb{E}$$\gets$$[]$, $i$$\gets$$i$$+$$1$, $\mathcal{F}[i]$$\gets$$+$$\infty$\label{line:ep:while:for:if4:eprestructure}
		
		
		\STATE\quad\quad\textbf{end if}\label{line:ep:while:for:if4:endif}
		\STATE\quad \textbf{end for}\label{line:ep:while:for:endFor} 
		
		\STATE\quad $\mathbb{T}\gets\,$EP-Reduce$(\mathbb{T},\mathbb{E})$\label{line:ep:while:epRestructure}, $\,$\textit{If $i\,=\,0$ then} \textbf{break}

		\STATE\quad$\alpha\gets \mathcal{F}_{R}$, $\,\mathcal{F}_{R}\gets\,$Find($\alpha,\,\mathcal{F}[j]-1$)\label{line:ep:while:alpha} \hfill \COMMENT{// Supposing $\mathcal{F}[j]\leq$ $\mathcal{F}[k]$, $\forall k\in[0,i]$}
		\STATE \quad$\mathcal{F}_{C}\gets\,$Find($\alpha,\mathcal{F}_{R}$$+$$1$), $\,$$\mathcal{F}_{CC}\gets\,$Find($\alpha,\mathcal{F}_{C}$$+$$1$)\label{line:ep:while:find}
		
		\STATE \quad$\forall i\in (\alpha,\mathcal{F}_{R}]$,  $u\gets\,$BON$[i]$ and $\mathbb{T},E_\mathbb{T}\gets\,$vPrune$(u)$\label{line:ep:while:prune}
		
		\STATE \quad $E_R,E_{next},f,t\gets$ErPrune($\mathcal{F}_{C},f,t,E_{next}$)\label{line:ep:while:PruneCheck}\label{line:ep:while:erPrune}
		\STATE \textbf{end while}\label{line:ep:while:end}\label{line:ep:while:endWhile}
		\STATE \textbf{return} Reset($\mathbb{T}, E_\mathbb{T}, E_i,E_o$)\label{line:ep:Reset}\label{line:ep:return}
		
	\end{algorithmic}
\end{algorithm}

\setcounter{algorithm}{0}
\begin{algorithm}[t]
	\floatname{algorithm}{Procedure}
	\caption{EP-Reduce($\mathbb{T},\mathbb{E}$)}
	\label{algo:ep_reduce}
	\begin{algorithmic}[1]
		\renewcommand{\algorithmicrequire}{\textbf{Input:}}
		\renewcommand{\algorithmicensure}{\textbf{Output:}}
		\renewcommand{\algorithmiccomment}[1]{  #1}
		
		
		\STATE $\mathcal{I}_f\gets \mathcal{F}_{R}+1$,  $\mathcal{I}_b\gets \mathcal{F}_{C}+1$, $\mathbb{G}\gets\big(V(\mathbb{T}),E(\mathbb{T})\cup \mathbb{E}\big)$\label{line:imp:indexes}
		\STATE \textbf{while} $\mathcal{I}_f\neq \mathcal{I}_b$ \textbf{do}\label{line:imp:while:start}
		
		\STATE \quad $s\gets\, $BON$[\mathcal{I}_f]$, $\,s.\mathcal{B}\gets\mathcal{I}_f$, $\,\mathcal{I}_f\gets\mathcal{I}_f+1$ \label{line:imp:while:s}
		
		\COMMENT{$\,\,\,\,\,$\quad//AddQ(v,s): BON$[\mathcal{I}_b]$$\gets$$ v$, $\mathcal{I}_b$$\gets$$\mathcal{I}_b$$+$$1$, $v.\mathcal{P}$$\gets$$ s$, \textit{Mark} $v$}
		
		\STATE\quad \textit{If $(s,v)\in E(\mathbb{G})$ \textbf{and} $v$ is unmarked then} AddQ$(v,s)$\label{line:imp:while:if}
		
		\STATE \textbf{end while}\label{line:imp:while:end}
		
		\STATE \textbf{for} the unmarked out-neighbors $u$ of $r$ in $\mathbb{G}$  \textbf{do}\label{line:imp:for:start}
		\STATE \quad \textit{If $u.\mathcal{B}\geq\mathcal{I}_f$ then} AddQ($u,r$) and SearchQ() \label{line:imp:for:if}
		
		\STATE \textbf{end for} \label{line:imp:for:end}
		\STATE  Reset\textit{ES}()\hfill\COMMENT{// $\mathbb{T}$ is an empty tree, and $\mathbb{E}=[]$}\label{line:imp:resetMemory}
		
		\STATE $\forall i\in [\mathcal{F}_{R}\negmedspace+\negmedspace1,\mathcal{I}_f)$, \textit{add} $(u.\mathcal{P},u)$ \textit{into} $\mathbb{T}$, \textit{if} $u$$\,=\,$BON$[i]$\label{line:imp:resetTree}
		\STATE \textbf{return} $\mathbb{T}$\label{line:imp:return}
		
	\end{algorithmic}
\end{algorithm}

Different from EE-BFS and EB-BFS, given an input graph $G$, EP-BFS first scans $G$ and divides $E(G)$ into $E_1,E_2,\dots,E_k$ as shown in Line~\ref{line:pg:ScanG}, Algorithm~\ref{algo:ep_bfs}, where $k=\lceil\frac{m}{n+Kn}\rceil$ if $(1+K)n$ edges can be maintained in the main memory. Besides, $E_1$ contains the first to the $\big((1+K)n\big)$th edges of $G$, and for any $i\in[2,k]$, $E_i$ contains the $\big((i-1)\times (1+K)n\big)$th to the $\big(i\times (1+K)n\big)$th edges of $G$. The edges in $E_1,E_2,\dots,E_k$ are maintained in the form of an adjacency list.

Then, EP-BFS utilizes the edges in $E_1,E_2,\dots,E_k$ to obtain $V_u$ for each node $u$ in $V(G)$, and to initialize $\mathbb{T}$ and the attributes of each node in $\mathbb{T}$, as demonstrated in Lines~\ref{line:pg:for:start}-\ref{line:pg:for:end}. Here, for a node $u$ of $G$, assuming $v$ is an out-neighbor of $u$ in $E_1,E_2,\dots, E_k$. Function Read adds $(u,v)$ into $E_o$ when the out-degree of $v$ in $G$ is zero; otherwise it adds $v$ into $V_u$, as shown in Line~\ref{line:pg:for:read}. When the in-degree of $u$ in $G$ is zero~(denoted by $id(u)=0$ in Line~\ref{line:pg:for:if1}), then EP-BFS enlarges $E_i$ with edges $(u,v)$~($v\in V_u$); otherwise EP-BFS enlarges $E_R$ and $\mathbb{E}$ with these edges, as shown in Line~\ref{line:pg:for:otherwise}. When $\mathbb{E}$ cannot be enlarged anymore, that is, the sum of the sizes of $\mathbb{E}$ and $\mathbb{T}$ reaches $(1+K)n$, EP-BFS executes function Reduce to restructure $\mathbb{T}$ into the BFS forest of the graph composed of $\mathbb{T}$ and $\mathbb{E}$, as illustrated in Line~\ref{line:pg:for:reduce}. 

After the computation of $V_u$, EP-BFS executes function TreeReduce~(as shown in Line~\ref{line:pg:Treduce}) to compute the BFS-tree of the graph composed of $\mathbb{T}$ and $\mathbb{E}$ by letting the dummy node $r$ be the root of $\mathbb{T}$, in which EP-BFS also initializes all the attributes  for the nodes on $\mathbb{T}$. Since in Line~\ref{line:pg:Treduce} $\mathbb{T}$ is reduced to a spanning tree of $G$, the initialized attributes $u.\mathcal{B}$ and $u.\mathcal{P}$ for a node $u$ in $\mathbb{T}$ represent $bfo(u,\mathbb{T})$ and the parent node of $u$ on $\mathbb{T}$, respectively.

The iteration of EP-BFS starts from Line~\ref{line:ep:while:start}, where EP-BFS iteratively scans $E_R$ rather than $E(G)$. For each scanned edge $e=(u,v)$, $e$ cannot be used to enlarge $\mathbb{E}$ when (\romannumeral1) $u.\mathcal{B}$ is no more than $\mathcal{F}_{R}$ or $v.\mathcal{B}$ is no more than $\mathcal{F}_{C}$ then $e$ is skipped directly\footnote{Initially, variable $\mathcal{F}_{R}=-\infty$, $\mathcal{F}_{C}=-\infty$ and $\mathcal{F}_{CC}=-\infty$, and the values of $\mathcal{F}_{R}$, $\mathcal{F}_{C}$ and $\mathcal{F}_{CC}$ are updated at the end of each iteration. In EP-BFS, we introduce notation $\mathcal{F}_{R}$ to record the minimum value among $\mathcal{F}[0],\dots,\mathcal{F}[i]$ in the previous iteration. Assuming $\mathcal{F}_R=bfo(u,T)$ and $\mathcal{F}_{C}=bfo(v,T)$, notation $\mathcal{F}_C$ and notation $\mathcal{F}_{CC}$, normally, refer to the breadth-first order of the rightmost child of $u$ on $T$ and that of $v$ on $T$, respectively.}, as shown in Line~\ref{line:ep:while:for:if1_GOTOL22}; (\romannumeral2) $v.\mathcal{P} = u$, as shown in Line~\ref{line:ep:while:for:if3:if}; (\romannumeral3) when $u.\mathcal{B}$ is equal to or larger than $\mathcal{F}[i]$ but $v.\mathcal{B}$ is smaller than $\mathcal{F}[i]$, as shown in Lines~\ref{line:ep:while:for:if3_insertFlage}-\ref{line:ep:while:for:if3:if}; (\romannumeral4) $u.\mathcal{B}<\mathcal{F}[i]$ and $e$ is not a V-BFS edge as classified by the tree composed of $\mathbb{T}$ and $E_\mathbb{T}$, as shown in Lines~\ref{line:ep:while:for:if3:elseif}-\ref{line:ep:while:for:if3:endif}.

In other words, when $\mathcal{F}[i]$ is no more than both $u.\mathcal{B}$ and $v.\mathcal{B}$, edge $(u,v)$ is added to $\mathbb{E}$ by Line~\ref{line:ep:while:for:if3:if}. When $u.\mathcal{B}$ is smaller than $\mathcal{F}[i]$, EP-BFS needs to check whether $(u,v)$ is a V-BFS edge as classified by the spanning tree of $G$ composed of $\mathbb{T}$ and $E_\mathbb{T}$ in Line~\ref{line:ep:while:for:if3:elseif}.  When $(u,v)$ is a V-BFS edge, EP-BFS not only adds $(u,v)$ to $\mathbb{E}$, but also updates the value of $\mathcal{F}[i]$ to the value of $u.\mathcal{B}$, as demonstrated in Line~\ref{line:ep:while:for:if3:elseif_process}. When $\mathbb{E}$ cannot be enlarged any further, i.e. the number of edges contained in $\mathbb{E}$ and $\mathbb{T}$ reaches $(1+K)n$, the IMP of EP-BFS, named EP-Reduce, is executed to reduce $\mathbb{T}$, after which $\mathbb{E}$ is set to $[]$, $i$ is set to $i+1$, and $\mathcal{F}[i]$ is set to $+\infty$, as illustrated in Lines~\ref{line:ep:while:for:if4}-\ref{line:ep:while:for:if4:endif}.

\textit{EP-Reduce~(the IMP of EP-BFS, as shown in Procedure~\ref{algo:ep_reduce}).} For efficiency,  EP-Reduce relies on a global array of length $n$, named BON, to maintain a total order of the unpruned nodes, i.e. a breadth-first order of the tree  composed of $\mathbb{T}$ and $E_\mathbb{T}$.

Based on BON, EP-Reduce can gradually delete all the edges contained in the in-memory sketch, and reset \textit{ES}~(which is introduced in the overview part of this section) for high efficiency. Specifically, in EP-Reduce, two variables $\mathcal{I}_f$ and $\mathcal{I}_b$ are utilized in order to implement a BFS queue $Q$ based on the array BON, which are initialized to $\mathcal{F}_{R}+1$ and $\mathcal{F}_{C}+1$, respectively, as demonstrated in Line~\ref{line:imp:indexes}, Procedure~\ref{algo:ep_reduce}. That is, at the beginning of EP-Reduce, $Q$ contains $\mathcal{I}_b-\mathcal{I}_f$ elements, where the first element of $Q$ is the $\mathcal{I}_f$th element in the array BON, and the last element of $Q$ is the $(\mathcal{I}_b-1)$th element in the array BON. 

Then, EP-Reduce searches until $Q$ is empty, i.e. $\mathcal{I}_f=\mathcal{I}_b$, as shown in Line~\ref{line:imp:while:start}. In the searching process, EP-Reduce removes the first node $s$ from $Q$, i.e. $s$ is the $\mathcal{I}_f$th element of array BON, and sets $s.\mathcal{B}$ to $\mathcal{I}_f$, as demonstrated in Line~\ref{line:imp:while:s}. After that, $\mathcal{I}_f$ is increased, and EP-Reduce visits the out-neighborhood of $s$. For each out-neighbor $v$ of $s$, if $v$ is unmarked, it executes the function AddQ which (\romannumeral1) adds $v$ into $Q$ by letting BON$[\mathcal{I}_b]$ and $\mathcal{I}_b$ be $v$ and $\mathcal{I}_b+1$, respectively, and (\romannumeral2) sets $v.\mathcal{P}$ to $s$ and marks $v$ as visited, as shown in Line~\ref{line:imp:while:if}. EP-Reduce also has to search all the unmarked out-neighbors $u$ of the dummy node $r$ with the function SearchQ, as illustrated in Lines~\ref{line:imp:for:start}-\ref{line:imp:for:end}. Note that, when $u.\mathcal{B}$ is less than $\mathcal{I}_f$, $u$ is marked, as demonstrated in Line~\ref{line:imp:for:if}. The operation of function SearchQ is the same as that shown in Lines~\ref{line:imp:while:start}-\ref{line:imp:while:end}.

Importantly, during the above search process, after visiting the out-neighborhood of a node $s$, all the outgoing edges of $s$ are no longer maintained in the main memory until EP-Reduce resets \textit{ES} in Line~\ref{line:imp:resetMemory}. That operation is crucial for the efficiency of EP-BFS since EP-Reduce can store the edges of the restructured $\mathbb{T}$ in memory in a specific order. To be precise, at the end of EP-Reduce, it scans the nodes in array BON from the $(\mathcal{F}_{R}+1)$th to the last, and for each scanned node $u$, it lets $u$ be the rightmost child of $u.\mathcal{P}$ in $\mathbb{T}$, as illustrated in Line~\ref{line:imp:resetTree}.

After each pass that EP-BFS scans $E_R$, EP-BFS executes its IMP if $\mathbb{E}$ is not empty, and checks whether this iteration is the last iteration, as shown in Line~\ref{line:ep:while:epRestructure}, Algorithm~\ref{algo:ep_bfs}. If it is not the last iteration, it executes Lines~\ref{line:ep:while:alpha}-\ref{line:ep:while:erPrune}.

Line~\ref{line:ep:while:alpha} uses $\alpha$ to denote $\mathcal{F}_{R}$ that is computed in the previous iteration, and uses $\alpha$ to obtain the new value of $\mathcal{F}_{R}$ by function Find. Function Find has two parameters, where one is $\alpha$ and the other is $\beta$.  It searches from $\beta$ to $\alpha+1$ and returns $\gamma$ if (\romannumeral1) there is a concrete number $\mathcal{C}$ in the range of $(\alpha,\beta]$ such that the $\mathcal{C}$th element of array BON is not a leaf node of $\mathbb{T}$, and (\romannumeral2) assuming $x$ is the rightmost child of BON$[\mathcal{C}]$, $\gamma=x.\mathcal{B}$. Otherwise, it returns $\beta$. To obtain $\mathcal{F}_{R}$, EP-BFS sets $\beta$ to be $\mathcal{F}[j]-1$, where $\mathcal{F}[j]$ refers to the smallest value among $\mathcal{F}[0]$, $\mathcal{F}[1]$, $\dots$, $\mathcal{F}[i]$. Line~\ref{line:ep:while:find} sets $\beta$ to $\mathcal{F}_{R}+1$ in order to find $\mathcal{F}_{C}$ by function Find and sets $\beta$ to $\mathcal{F}_{C}+1$ in order to find $\mathcal{F}_{CC}$. Line~\ref{line:ep:while:prune}~(function vPrune) removes all the edges $e_p=(u,v)$ related to $u$ from $\mathbb{T}$ and adds $e_p$ into $E_\mathbb{T}$, where $u=\,$BON$[i]$ assuming $i$ is a number in the range of $(\alpha,\mathcal{F}_{R}]$.

Line~\ref{line:ep:while:for:if2_englargeE} and Line~\ref{line:ep:while:erPrune} are operations for reducing $E_R$ to $E_{next}$. $E_{next}$ is initialized to $\phi$, as shown in Line~\ref{line:ep:intializedOthers}. If in the $k$th iteration of EP-BFS, function ErPrune sets $E_{next}$ to $[]$, then in the $(k+1)$th iteration EP-BFS adds edges $e=(u,v)$ into $E_{next}$ by function Enlarge, if $u.\mathcal{B}$ is larger than $\mathcal{F}_{C}$ and $v.\mathcal{B}$ is larger than $\mathcal{F}_{CC}$. Function ErPrune sets $E_{next}$ to $[]$ iff $\mathcal{F}_{C}-f>n\times \gamma\times (\frac{m}{n})^t$, in which (\romannumeral1) $f$ and $t$ are variables that are initialized in Line~\ref{line:ep:intializedOthers} and updated in function ErPrune; (\romannumeral2) $\gamma$ is a parameter of EP-BFS which by default is set to $8\%$; (\romannumeral3) $f$ is used to record $\mathcal{F}_{C}$ when $E_{next}$ is set to $[]$; (\romannumeral4) $t$ is used to count the times that $E_{next}$ has been set to $[]$. That is, when $E_{next}$ is set to $[]$, ErPrune also has to set $f$ and $t$ to $\mathcal{F}_{C}$ and $t+1$, respectively.


Line~\ref{line:ep:return} returns $T$ with $\mathbb{T}$, $\mathbb{E_\mathbb{T}}$, $E_i$ and $E_o$, where $V(T)=V(G)$. When $G$ does not have a node $u$ whose in-degree or out-degree is zero, $E(T)=E(\mathbb{T})\cup E_\mathbb{T}$. Otherwise, if the in-degree of $u$ is zero, $u$ is added to $T$ as the rightmost child of $r$, by default. If the out-degree of $u$ is zero, $u$ is added to $T$ as the leftmost child of $r$, by default. When a specific requirement for the position of $u$ on $T$ exists, EP-BFS keeps $u$ in $V(\mathbb{T})$ and all the edges related to $u$ in $E_R$.

\subsection{Analysis} \label{sec:epBFs:analysis}
In this part, we first discuss the correctness of EP-BFS, and then discuss its time and I/O consumption. Figure~\ref{fig:ESM} is a demonstration of the relationship among the edge lists introduced below.

Assuming $E_1,E_2,\dots, E_k$ are $k$ sublists of $E$, where $E_1$ contains the first $|E_1|$ elements of $E$, $E_2$ covers the $(|E_1|+1)$th elements of $G$ to the $(|E_1|+|E_2|)$th elements of $G$. Notation $RE(\mathcal{E}_1,\mathcal{E}_2,\dots,\mathcal{E}_\delta)$ represents a function which restructures the spanning tree $T$ of $G$ with edge sets $\mathcal{E}_1,\mathcal{E}_2,\dots,\mathcal{E}_\delta$, where  $T_1$ denotes the input form of $T$.  $T_{j}$~($j\in[2,\delta+1]$) represents the BFS-tree of the graph composed of $T_{j-1}$ and $\mathcal{E}_{j-1}$ under Stipulation~\ref{sti:stipulation}. $T_{\delta+1}$ denotes the output form of $T$.

\begin{lemma}
	For any spanning tree $T$ of $G$, calling $RE(E_1,E_2,\dots, E_k)$ constantly can obtain a BFS-tree of $G$, and the function $RE(E_1,E_2,\dots, E_k)$ is executed at most \textit{LLSP}$(G)$ times in the worst case. 
\end{lemma}
\begin{proof}
	The correctness of this statement is obvious according to Theorem~\ref{theorem:eb_searchTime}, since one invocation of the function $RE(T)$ scans the entire $G$ and uses all the edges of $G$ to restructure $T$. 
\end{proof}

Supposing $E_i$ has $\alpha$ different V-BFS edges as classified by $T_i$ which are in the order of $e_1,e_2,\dots,e_{\alpha}$, where if $\forall \beta\in[1,\alpha]$, $e_\beta=(u_\beta,v_\beta)$, then  $bfo(u_1,T_i)\leq bfo(u_2,T_i)\leq \dots\leq bfo(u_\alpha,T_i)$. $S_i$ denotes a sublist of $E_i$, where the edges in $S_i$ are in the same order as they are in $E_i$. $S_i$ contains an edge $e_S$ of $E_i$ iff (\romannumeral1) $e_S$ is $e_1$ or is behind $e_1$ in $E_i$, (\romannumeral2) if $e_S=(u_S,v_S)$, then $u_S$ is not the parent node of $v_S$ on $T_i$, and (\romannumeral3) both $bfo(u_S,T_i)$ and $bfo(v_S,T_i)$ are not smaller than $bfo(u_1,T_i)$. Below, we use the notation \textit{Min}$(S_i)$ to denote $bfo(u_1,T_i)$.

\begin{lemma}
	For any spanning tree $T$ of $G$, calling $RE(S_1,S_2,\dots, S_k)$ constantly can obtain a BFS-tree of $G$, and function $RE(S_1,S_2,\dots, S_k)$ is executed at most \textit{LLSP}$(G)$ times in the worst case. 
\end{lemma}
\begin{proof}
	The statement is definitely valid, since the edges in $E_i$ that are not contained in $S_i$ do not affect the form of $T_{i+1}$ under Stipulation~\ref{sti:stipulation}.
\end{proof}

\begin{figure}[t]
	\centering
	\includegraphics[scale = 0.725]{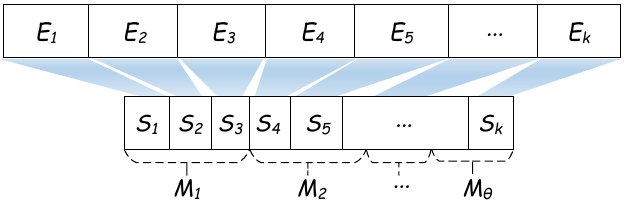}
	\caption{A schematic view of the relationship among the edge lists $E_1,E_2,\dots,$ $E_k$, $S_1,S_2,\dots,S_k$ and $M_1,M_2,\dots,M_\theta$.}
	\label{fig:ESM}
\end{figure} 

Letting $M_1,M_2,\dots,M_\theta$ be the edge lists related to $S_1,$ $S_2,\dots,S_k$, where $\theta\leq k$. $M_1$ is the concatenation of the edge lists from $S_1$ to $S_{c(1)}$ in which $1\leq c(1)\leq k$. $M_2$ is the concatenation of the edge lists from $S_{c(1)}$ to $S_{c(2)}$ in which $c(1)< c(2)\leq k$... $M_\theta$ is the concatenation of the edge lists from $S_{c(\theta-1)}$ to $S_{c(\theta)}$ in which $c(\theta-1)< c(\theta) =\theta$. 

\begin{lemma}
	For any spanning tree $T$ of $G$, calling $RE(M_1,M_2,...,M_\theta)$ constantly can obtain a BFS-tree of $G$, and function $RE(M_1,M_2,...,M_\theta)$ is executed at most \textit{LLSP}$(G)$ times in the worst case.
\end{lemma}
\begin{proof}
	Assuming $M_i$ is the concatenation of $S_{c(i-1)+1},$ $S_{c(i-1)+2},\dots,S_{c(i)}$. $RE(M_1,M_2,...,M_\theta)$ restructures $T_i$ to $T_{i+1}$ with $M_i$. $RE(S_1,S_2,\dots, S_k)$ restructures $T_{c(i-1)+1}$ to $T'_{c(i)+1}$ with $S_{c(i-1)+1},$ $S_{c(i-1)+2},\dots,S_{c(i)}$. Then, the statement is valid, since $T'_{c(i)+1}$ is equivalent to $T_{i+1}$. 
\end{proof}

\setcounter{theorem}{3}
\begin{theorem}
	\label{theorem:ep:correctness_EPbfs}
	EP-BFS returns $T$ correctly with at most \textit{LLSP}$(G)$ iterations.
\end{theorem}
\begin{proof}
	The statement is proved with two parts. The former part proves that when no node is reduced after scanning $G$, EP-BFS returns $T$ correctly with at most \textit{LLSP}$(G)$ iterations. The latter part proves that EP-BFS is also correct if certain nodes are pruned after scanning $G$.
	
	(1) Assuming in Line~\ref{line:ep:intializedOthers}, Algorithm~\ref{algo:ep_bfs}, $V(\mathbb{T})=V(G)$. This statement is correct, iff EP-BFS restructures $T$ with $M_1,M_2,...,M_\theta$ and EP-Reduce restructures $T$ correctly, where in EP-BFS $T$ is decomposed to $\mathbb{T}$ and $E_\mathbb{T}$. The correctness of EP-Reduce is obvious, since the array BON correctly maintains the breadth-first order of the nodes on the restructured $T$. EP-BFS does restructure $T$ with $M_1,M_2,...,M_\theta$, which is detailedly discussed below. Firstly, at the end of each iteration, $\mathcal{F}_{R}$ is updated to the minimum number of \textit{Min}$(S_i)$ where $i$ is in the range of $[1,k]$. Secondly, $\mathcal{F}_{C}$ is updated to $bfo(v,T)$ if $v$ has the largest breadth-first order on $T$ when requiring that the breadth-first order of the parent of $v$ on $T$ is no larger than $bfo(u,T)$ on $T$, in which we assume \textit{Min}$(S_i)$ refers to $bfo(u,T)$. Besides, $\mathcal{F}_{CC}$ is also updated to $bfo(v,T)$ under the assumption that $\mathcal{F}_{C}$ refers to $bfo(u,T)$.
	
	(2) Assuming in Line~\ref{line:ep:intializedOthers}, Algorithm~\ref{algo:ep_bfs}, $V(\mathbb{T})\neq V(G)$. EP-BFS also returns the correct $T$ because a node $u$ is pruned from $G$ after scanning $G$ once, iff, $u$ has zero in-degree~(denoted by $id(u)=0$) or zero out-degree~(denoted by $od(u)=0$). When $id(u)=0$, EP-BFS adds $u$ into $T$ as the rightmost child of $r$, by default. When $od(u)=0$, EP-BFS adds $u$ into $T$ as the leftmost child of $r$, by default. 
\end{proof}

Based on Theorem~\ref{theorem:ep:correctness_EPbfs}, the I/O and time costs of EP-BFS in the worst case are $O\big(m\times\,$\textit{LLSP}$(G)\big)$ and $O\big(\frac{m}{B}\times $\textit{LLSP}$(G)\big)$, respectively. This is because (\romannumeral1) the time cost of EP-Reduce is $O(n)$, and (\romannumeral2) EP-BFS calls EP-Reduce at most $\lceil\frac{m}{n}\rceil$ times in one iteration in the worst case.

\begin{table*}[t]
	\caption{The experimental results on the real datasets, where  ``-'' indicates that a semi-external BFS algorithm is timed out on a dataset, and the time limit is $24$ hours. The running time~(RT) is reported in seconds. The amount of DT, the size of $G$ and the size of $\mathcal{A}$ are reported in gigabytes.}
	\label{tab:real_datasets}
	\small
	\begin{center}
		\setlength{\tabcolsep}{1.9mm}{
			\begin{tabular}{c|ccc|cc|cc|cc}
				\hline \multirow{2}*{Dataset} & \multirow{2}*{$n/10^6$} &\multirow{2}*{$m/10^6$}&\multirow{2}*{$m/n$} &
				
				\multicolumn{2}{c|}{EB-BFS} &\multicolumn{2}{c|}{EP-BFS}& 
				\multicolumn{2}{c}{Size~(GB)}
				\\
				
				& & & &		RT~(s) & I/O~(GB)& RT~(s) & I/O~(GB)  & $G$& $\mathcal{A}$~($K=1$)\\
				\hline
				hollywood-2011 	&$2.18$&$229$&$105$	&$260$&$10.2$&$60$&$6.5$&$1.71$&$0.089$\\
				eu-2015-host  &$11.3$&$387$	&$34.4$	&$837$&$17.3$&$104$&$10.4$&$2.88$&$0.35$\\
				uk-2002	 	&$18.5$&$298$	&$16.1$
				&$4,097$&$75.5$&$138$&$12.7$&$2.22$&$0.57$\\	
				gsh-2015-tpd&$30.8$&$602$	&$19.5$ 
				&$1,711$&$26.9$&$214$&$14.2$&$4.49$&$0.93$\\
				it-2004		&$41.3$&$1,151$&$27.9$ 
				&$6,541$&$111$&$647$&$54.0$&$8.57$&$1.25$\\
				twitter-2010	&$41.7$&$1,468$ & $35.3$
				&$11,165$&$175$&$682$&$40.8$&$10.9$&$1.26$\\
				sk-2005		&$50.6$&$1,949$&$38.5$ 
				&$13,414$&$218$&$839$&$75.3$
				&	$14.5$&$1.53$\\	
				Friendster&$68.3$&$2,586$&$37.8$&$11,650$&$154$&$1,217$&$59.0$&$19.3$&$2.06$\\
				uk-2007-05	&$105$&$3,738$&$35.3$ 
				&-&-&$1,707$&$141$&	$27.9$&$3.17$\\
				uk-2014&$787$&$47,614$&$60.4$
				&-&-&$40,155$&$2,303$&$355$&$23.6$\\
				clueweb12& $978$&$42,574$&$43.5$ 
				&-&-&$31,384$&$1,856$&$317$&$29.4$\\	
				gsh-2015&$988$&$33,877$&$34.3$&-&-&$25,535$&$1,588$&$252$&$29.7$\\
				eu-2015  &$1,071$&$91,792$	&$85.7$
				
				&-&-&$73,526$&$4,124$		&$683$&$32.2$\\
				WDC-2014&$1,727$&$64,422$&$37.3$ 
				&-&-		&$32,496$&$1,558$&	$480$ &$51.7$\\
				
				\hline				
				
			\end{tabular}
		}
	\end{center}
\end{table*}

\section{Performance Evaluation}\label{sec:experiments}
In this section, we evaluate the performance of the proposed efficient semi-external BFS algorithm EP-BFS against the two naive algorithms EE-BFS and EB-BFS for the semi-external BFS problem and a single-machine out-of-core graph processing system GridGraph~\cite{DBLP:conf/usenix/ZhuHC15}, on both synthetic and real graphs. Comparing EP-BFS with EE-BFS, EB-BFS and GridGraph is sufficient, because as we mention in Section~\ref{sec:related_work}, there is no semi-external BFS algorithm that has been proposed before this paper because of the difficulty of the problem. Besides, it is reported~\cite{DBLP:conf/asplos/ZhangWZQHC18,DBLP:conf/usenix/AiZWQCZ17,DBLP:conf/usenix/Vora19} that (\romannumeral1) current single-machine out-of-core graph processing systems are efficient and support various graph problems including BFS; (\romannumeral2) GridGraph exhibits relatively stable performance under varying memory conditions~\cite{DBLP:conf/asplos/ZhangWZQHC18,DBLP:conf/usenix/AiZWQCZ17,DBLP:conf/usenix/Vora19}.


All the experiments are executed on a DELL Optiplex Tower Plus 7010 (Intel(R) Core(TM) i9-13900 CPU @ 2.00GHz + 64GB memory + 64bit windows 11). Besides, in our experiments, all the semi-external algorithms~(EE-BFS, EB-BFS and EP-BFS) are implemented in Java with jdk-8u112-windows-x64. GridGraph is implemented in C++ accessed from ``\textit{https://github.com/colbacc8/GraphPreprocessing}''. We limit each experiment to $24$ hours. \textit{As EE-BFS cannot be terminated within our time limit in almost every experiment, the experimental results of EE-BFS are not demonstrated and discussed in the following of this section.} The storage devices used in our experiment include a $5$TB HDD  and a  $2$TB SSD\footnote{For reproducibility, the HDD and SSD used in this section is obtained from ``\textit{https://item.jd.com/100021394398.html}'' and ``\textit{https://item.jd.com/100} \textit{007541860.html}'', respectively}. All the experiments are conducted on the HDD by default.

The efficiency and the number of I/Os are the main factors that we are interested in for each semi-external BFS algorithm on each graph. The former is measured by the running time, whereas the latter is measured by the total size of disk accesses, i.e. the amount of data transferred~(DT). In addition, the effect of the size of internal memory, i.e. the size of $K$, on the performance of each algorithm is also evaluated in the experiments. For each experiment, we also report the size of the in-memory sketch of EP-BFS.

To evaluate all the proposed algorithms, various large-scale datasets including both synthetic and real datasets are utilized. Table~\ref{tab:real_datasets} summarizes all the real datasets, in which dataset Friendster can be accessed from the website ``\textit{http://www.konect.cc/}'', dataset WDC-2014 can be accessed from the website ``\textit{http://www.webdatacommons.org/}'', and all the other datasets can be accessed from the website ``\textit{http://law.di.unimi.it/}''. The datasets are detailedly introduced on the websites where they can be accessed. Synthetic datasets are randomly generated~\cite{DBLP:journals/csur/DrobyshevskiyT20}, according to Erdös-Rényi~(ER) model. To be precise, for a synthetic dataset $G=(V,E)$, we randomly and repeatedly generate an edge $e=(u,v)$ where $u,v\in V$ and $u\neq v$, and the edges in $E$ are unique. \textit{$|V|$ of a generated synthetic dataset $G$ is in the range of $[5$ million, $2$ billion$]$, while $|E|$ is in the range of $[15$ million, $40$ billion$]$.}

One experiment of this section evaluates one semi-external algorithm on a disk-resident general directed graph $G$, where the edges of $G$ are stored in a random order. Its in-memory sketch $\mathcal{A}$ is implemented by the two arrays \textit{ES} and \textit{Adj}, as introduced in Section~\ref{sec:epBFs}. The available memory space maintains $(K+1)n$ edges for $\mathcal{A}$, where $K$ is in the range of $[0.05,2]$ in the experiments, and by default $K=1$. This implementation method can protect all the implemented algorithms from being affected by the garbage collection mechanisms of implementation languages, since we only modify the values of the arrays that are initialized at the beginning. The experiments conducted with Java and Eclipse IDE are  controled by six parameters: \textit{-Xms}, \textit{-Xmx}, \textit{-Xmn}, \textit{-XX:PretenureSizeThreshold}, \textit{-XX:-UseGCOverheadLimit} and \textit{-XX:+PrintGCDetails}\footnote{For more details, please see ``\textit{https://www.oracle.com/}''.}.

\subsection{The impact of graph structure on real graphs}
Table~\ref{tab:real_datasets} presents the experimental results of evaluating the performance of EB-BFS and that of EP-BFS on $14$ real graphs with different structures, as briefly introduced below. hollywood-2011 has the largest average node degree.   eu-2015-host is the host graph of graph eu-2015.  gsh-2015-tpd contains the top private domains of Graph gsh-2015. twitter-2010 and Friendster are two large social networks that can be openly accessed. uk-2007-05 is a relatively large graph containing millions of nodes and billions of edges. clueweb12, gsh-2015 and graph eu-2015 all include about $1$ billion nodes where graph eu-2015 has over $91$ billion edges. WDC-2014 is a massive hyperlink graph that contains over $1.7$ billion nodes.  These graphs are generated from different fields in different ways, where (\romannumeral1) some graphs are social graphs, e.g., hollywood-2011 contains the relationship between actors and movies; (\romannumeral1) some graphs are crawls from the internet in different domains, e.g. uk-2007-05 is collected from the .uk domain by DELIS project~\cite{BSVLTAG}.

Experimental results show that on the graph with millions of nodes or tens of millions of nodes, EB-BFS can work but is inefficient compared with EP-BFS, as shown in Table~\ref{tab:real_datasets}. Specifically, on the graph hollywood, EB-BFS consumes $260$ seconds and $10.2$GB I/Os, while EP-BFS only needs $60$ seconds and $6.5$GB I/Os. On the graph uk-2002, the time cost of EP-BFS~($138$ seconds) is $3$ percent of that of EB-BFS~($4,097$ seconds), and the I/O cost of EP-BFS~($12.7$GB) is less than $17$ percent of that of  EB-BFS~($75.5$GB). On the graph it-2004, EP-BFS is an order of magnitude faster than EB-BFS, and the I/O consumption of EP-BFS is only half of that of EB-BFS. On the large-scale graph sk-2005, EP-BFS can be accomplished with $1,217$ seconds and $59$GB, while EB-BFS requires $11,650$ seconds and $154$GB. On the social graph twitter-2010 and Friendster, the performance of EP-BFS is also significantly better than that of EB-BFS. 

The reason why EP-BFS is much more efficient than EB-BFS on all the above graphs $G$ is that EP-BFS, for one thing, can prune a great number of nodes from $G$ after scanning $G$ once since real graphs contain many nodes whose  in-degree or out-degree is $0$. For example, the graph uk-2002 has over $10^3$ nodes whose in-degree is zero, and over $10^6$ nodes whose out-degree is zero, as illustrated on the website ``\textit{https://law.di.unimi.it/webdata/uk-2002/}''. For another, EP-BFS is a cache-friendly algorithm, which can fully utilize cache for high efficiency. Besides, by introducing $\mathcal{F}[i]$, EP-BFS can prune a lot of the newly scanned edges during the in-memory sketch enlarging process. For instance, the I/O cost of EP-BFS is almost half of that of EB-BFS on the graph sk-2005, but the time cost of EP-BFS is only less than $7$ percent of that of EB-BFS on the graph sk-2005.

On the massive graphs with hundreds of millions of nodes or billions of nodes, EB-BFS cannot be finished within our time limit, whereas EP-BFS is still efficient, as demonstrated in Table~\ref{tab:real_datasets}. To be precise, EP-BFS can obtain the BFS results within half an hour, but EB-BFS cannot within $24$ hours. The time costs of EP-BFS on graph uk-2014, clueweb12, gsh-2015, eu-2015 and WDC-2014 are about $11.2$ hours, $8.7$ hours, $7.1$ hours, $20.4$ hours and $9.1$ hours, respectively; the I/O costs of EP-BFS on these massive graphs are about $2.25$TB, $1.8$TB, $1.6$TB, $4.0$TB and $1.5$TB, respectively. The above results show that the proposed efficient semi-external BFS algorithm, EP-BFS, can process the massive graphs with billions of nodes or tens of billions of edges within an acceptable time frame. 

Table~\ref{tab:real_datasets} also reports the size of the in-memory sketch $\mathcal{A}$ of EP-BFS when it processes the $14$ real graphs utilized in the experiments of this part. For example, the size of $\mathcal{A}$ on graph uk-2002 is $0.57$GB, while the size of the entire graph uk-2002 is $2.22$GB. The size of $\mathcal{A}$ on graph Friendster is only $10\%$ of the size of the entire Friendster. On the massive graph eu-2015, the size of $\mathcal{A}$ is only $32.2$GB, where the size of graph eu-2015 is $683$GB, and $\frac{32.2}{683}\approx4.7\%$. Applications that need to process such massive graphs with memory constraint could consider implementing EP-BFS, while EB-BFS is enough when applications just need to process small graphs such as hollywood-2011, as the implementation of EB-BFS is simpler.

\begin{figure}[t]
	\centering
	\renewcommand{\thesubfigure}{}
	
	\subfigure[(a) Node statistic]{
		\includegraphics[scale = 0.85]{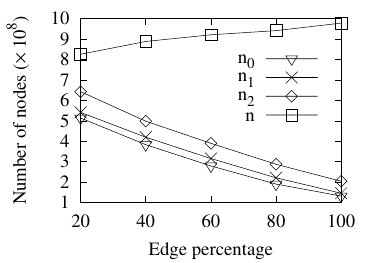}}
	\subfigure[(b) Edge statistic]{
		\includegraphics[scale = 0.85]{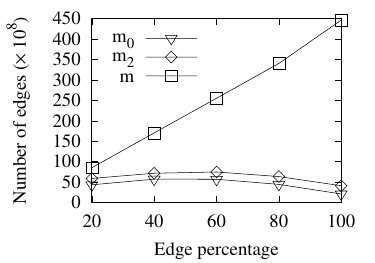}}
	\subfigure[(c) Efficiency]{
		\includegraphics[scale = 0.85]{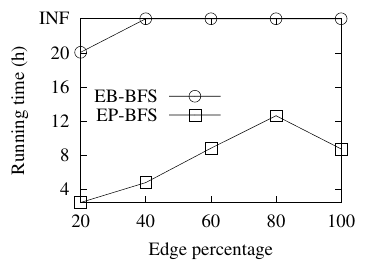}}
	\subfigure[(d) I/O]{
		\includegraphics[scale = 0.85]{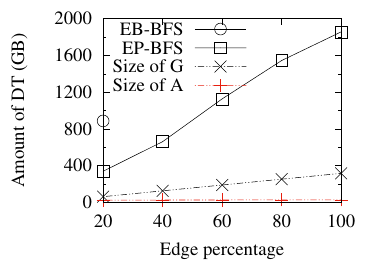}}
	\caption{Experiment of uniformly and randomly selecting edges from massive datasets clueweb12, in which subfigure~(a) and (b) are the node and edge statistics of the generated graphs $G$,  and subfigures~(c) and (d) report the experimental results. $n=|V(G)|$ and $m=|E(G)|$. $S$ is the largest SCC of $G$. $n_0$ is the number of nodes which are not strongly connected to any other nodes. $n_1$ is the number of nodes which are strongly connected to $[0,10]$ other nodes. $n_2$ is the number of nodes that are not in $S$. $m_0$ represents the number of edges, whose ends are not strongly connected. $m_1$ is the number of edges that are not in $S$.}
	\label{fig:uk-2014_clueweb1212_R}
\end{figure}

\subsection{The impact of randomly generating graphs from real massive graphs}\label{sec:experiments:cluewebSubGRAPHS}

In this part, experiments are conducted to evaluate the performance of the proposed algorithms on graphs that are generated by randomly selecting edges from the graph clueweb12. The reason for selecting graph clueweb12 over the others is that its average node degree and scale are relatively high. We vary the size of the selected edge percent from $20$ to $100$, as shown on the x-axes of Figure~\ref{fig:uk-2014_clueweb1212_R}. 

Specifically, to generate a graph $G_p$ from the graph clueweb12 (denoted by $G$) where $\frac{|E(G_p)|}{|E(G)|}=p$, the graph generation process scans the disk-resident $G$ once, and for each scanned edge $e$, the probability that $e$ is added to $G_p$ is $p$. Due to the fact that the size of $|E(G)|$ is huge, $|E(G_p)|=p\times m$~\cite{Statistic}. In Figure~\ref{fig:uk-2014_clueweb1212_R}(a) and Figure~\ref{fig:uk-2014_clueweb1212_R}(b), the number of nodes of $G_p$, the number of edges of $G_p$ and five other properties of $G_p$~(as illustrated in the title of Figure~\ref{fig:uk-2014_clueweb1212_R}) are reported to understand the structures of $G_p$. From the statistics of $G_p$, we can see that the generated graphs have different graph structures. For instance, when $p=20$, $\frac{|V(G_p)|}{|V(G)|}\approx 85\%$, the largest strongly connected component of $G_p$ contains only $25\%$ of the nodes of $G$, whereas the largest strongly connected component of $G$ includes about $80\%$ of the nodes of $G$. Based on the generated graphs, the performance of the proposed algorithms can be fully tested with reproducibility.

Experimental results are given in Figure~\ref{fig:uk-2014_clueweb1212_R}(c) and Figure~\ref{fig:uk-2014_clueweb1212_R}(d). In these experiments, EB-BFS can only obtain the BFS results when $p=20$, and it consumes over $20$ hours and almost $900$GB I/Os. However, the efficient semi-external BFS algorithm EP-BFS only needs less than $2.5$ hours and less than $350$GB I/Os when running on the graph that is generated by randomly selecting $20\%$ of the edges from the graph clueweb12. The I/O costs of EP-BFS increase linearly when the value of $p$ increases. The time consumption of EP-BFS increases with the number of nodes or the number of edges contained in the generated graphs increases, as shown in Figure~\ref{fig:uk-2014_clueweb1212_R}(a)-(c). However, it is noticed that EP-BFS on the generated graph with $p=80$ consumes more running time than it does on the entire clueweb12. In terms of the I/O cost of EP-BFS on these two graphs, randomly selecting edges from clueweb12 increases the total time that EP-BFS invokes its IMP, i.e. EP-Reduce. The experimental results further prove that EP-BFS is efficient for processing large-scale graphs with different structures in semi-external memory model.

\subsection{The comparison of the proposed algorithms with GridGraph}

In this part, we evaluate the performance of the proposed algorithms and that of the single-machine out-of-core graph processing system GridGraph for computing the total BFS order on the subgraphs of twitter-2010. The reason why we select twitter-2010 over the others is that this dataset is widely-used for evaluating the performance of the single-machine out-of-core graph processing systems~\cite{DBLP:conf/asplos/ZhangWZQHC18,DBLP:conf/usenix/AiZWQCZ17,DBLP:conf/usenix/Vora19}. The subgraphs of twitter-2010 utilized in this part are generated in the same way as the subgraphs of clueweb12 are generated in Section~\ref{sec:experiments:cluewebSubGRAPHS}. We vary the size of the selected edge percent from $0.1$ to $100$. However, GridGraph times out on the subgraphs with more than $3\%$ edges, which is noteworthy.  EP-BFS only requires $682$ seconds and $1.26$GB memory, while GridGraph is allowed to use $2$GB memory but is timed out.

The experimental results on the subgraphs with $0.1\%$ to $3\%$ edges are shown in Figure~\ref{fig:rangeTwitter_outCore}, where Figure~\ref{fig:rangeTwitter_outCore}(a) and Figure~\ref{fig:rangeTwitter_outCore}(b) demonstrate the structures of these subgraphs. Figure~\ref{fig:rangeTwitter_outCore}(c) reports the running time for each experiment of this part, and Figure~\ref{fig:rangeTwitter_outCore}(d) reports the memory usage of EP-BFS and that of GridGraph on each subgraph. Although larger subgraphs have a larger proportion of nodes connected, the time consumption of GridGraph is still growing exponentially, as it is related to numerous random disk I/Os. For instance, to compute the total BFS order of the nodes for the graph with $3\%$ edges, EP-BFS consumes $21$ seconds, EB-BFS consumes $87$ seconds, while GridGraph requires more than $22$ hours. The memory consumption of EP-BFS is only about a quarter of that of GridGraph.

\begin{figure}[t]
	\centering
	\renewcommand{\thesubfigure}{}
	
	\subfigure[(a) Node statistic]{
		\includegraphics[scale = 0.85]{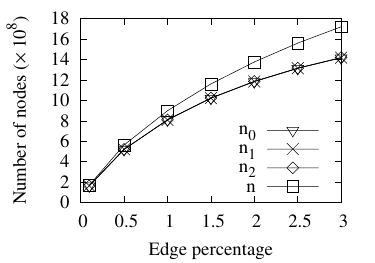}}
	\subfigure[(b) Edge statistic]{
		\includegraphics[scale = 0.85]{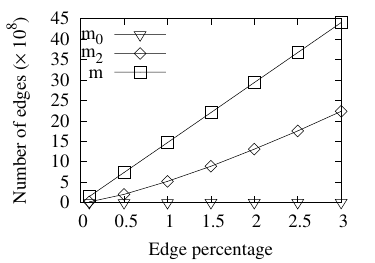}}
	\subfigure[(c) Efficiency]{
		\includegraphics[scale = 0.85]{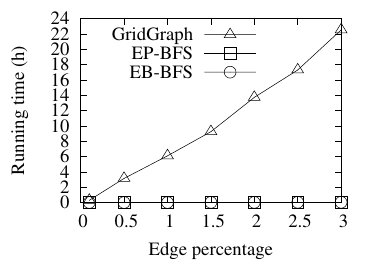}}
	\subfigure[(d) Memory usage]{
		\includegraphics[scale = 0.85]{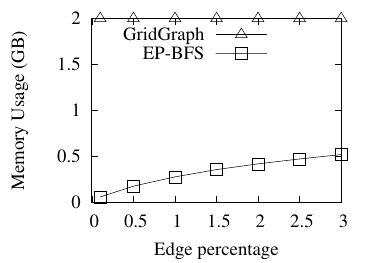}}
	\caption{The experimental results of EB-BFS, EP-BFS and GridGraph on the subgraphs of twitter-2010. The meanings of $n,n_0,n_1,n_2$ and $m,m_0,m_2$ are discussed in Figure~\ref{fig:uk-2014_clueweb1212_R}.}
	\label{fig:rangeTwitter_outCore}
\end{figure}

\begin{figure}[t]
	\centering
	\renewcommand{\thesubfigure}{}
	
	\subfigure[(a) Efficiency~(small step)]{
		\includegraphics[scale = 0.85]{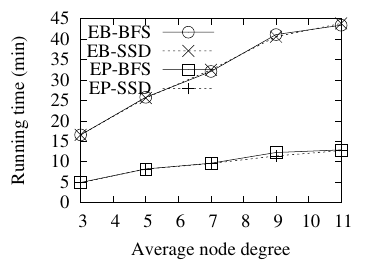}}
	\subfigure[(b) I/O~(small step)]{
		\includegraphics[scale = 0.85]{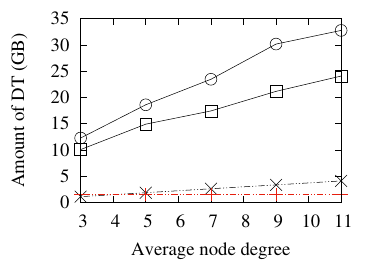}}
	\subfigure[(c) Efficiency~(big step)]{
		\includegraphics[scale = 0.85]{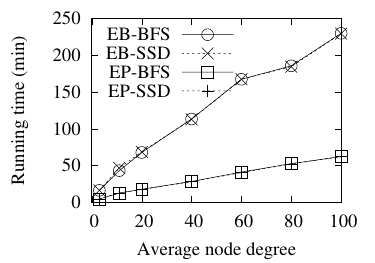}}
	\subfigure[(d)  I/O~(big step)]{
		\includegraphics[scale = 0.85]{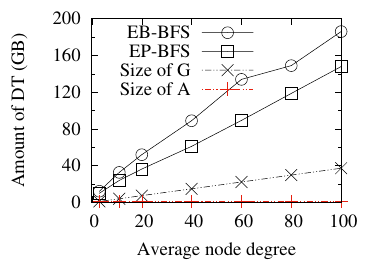}}
	\caption{The experimental results of varying average node degree on synthetic datasets in which $n=50$ million. The legend of subfigure (b) is omitted, which is the same as that of subfigure (d).}
	\label{fig:rangeD_S}
\end{figure}

\begin{figure}[t]
	\centering
	\renewcommand{\thesubfigure}{}
	\subfigure[(a) Efficiency]{
		\includegraphics[scale = 0.85]{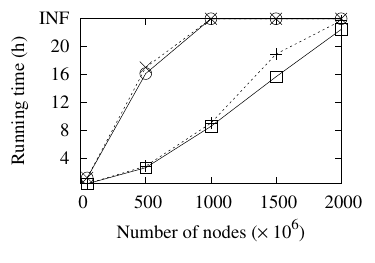}}
	\subfigure[(b) I/O]{
		\includegraphics[scale = 0.85]{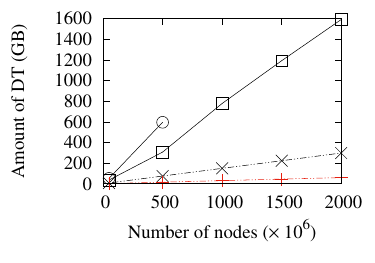}}
	
	\caption{The experimental results of varying $n$ on synthetic datasets by fixing $\frac{m}{n}=20$. The legends of subfigure (a) and (b) are omitted, which are the same as that of Figure~\ref{fig:rangeD_S}(c) and (d), respectively.}
	\label{fig:rangeV_S}
\end{figure}

\begin{figure}[t]
	\centering
	\renewcommand{\thesubfigure}{}
	\subfigure[(a) Efficiency]{
		\includegraphics[scale = 0.85]{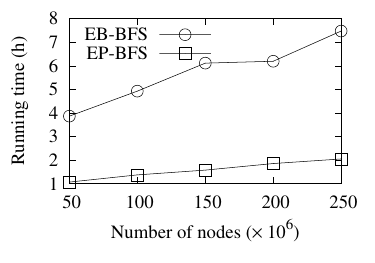}}
	\subfigure[(b) I/O]{
		\includegraphics[scale = 0.85]{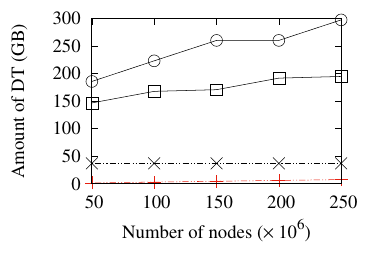}}
	
	\caption{The experimental results of varying average node degree on synthetic datasets in which $m=5$ billion.  The legend of subfigure (b) is omitted, which is  the same as that of Figure~\ref{fig:rangeD_S}(d).}
	\label{fig:rangeM_S}
\end{figure}

\subsection{The impact of varying $\frac{m}{n}$ on synthetic graphs with fixed $n$}\label{sec:experiments:AND}

In this part, we conduct two groups of experiments, as demonstrated in Figure~\ref{fig:rangeD_S}. The average node degree~($\frac{m}{n}$) in the first group of experiments ranges from $3$ to $11$, as illustrated in Figure~\ref{fig:rangeD_S}(a) and Figure~\ref{fig:rangeD_S}(b), whereas that in the second group ranges from $3$ to $100$, as depicted in Figure~\ref{fig:rangeD_S}(c) and Figure~\ref{fig:rangeD_S}(d). The node number of each generated graph is set to $50$ million. The experimental results demonstrate that the I/O costs of both EP-BFS and EB-BFS increase when the average node degree of the generated graph increases. The time consumption of EP-BFS on the generated graph with $\frac{m}{n}=100$ is about $50$ minutes, whereas the time consumption of EB-BFS on the generated graph with $\frac{m}{n}=11$ is already over $40$ minutes. The performance of EP-BFS is better than that of EB-BFS, especially when $\frac{m}{n}$ of the generated graph is high. The reason is obvious: when $\frac{m}{n}$ is high, EP-BFS prunes more edges with parameter $\mathcal{F}[i]$ and executes fewer times of its IMP, compared with EB-BFS. Besides, as the total number of edges that need to be scanned increases when $\frac{m}{n}$ of the generated graph increases, the cache efficiency of EP-BFS could greatly reduce the running time of EP-BFS.

\subsection{The impact of varying $n$ on synthetic graphs with fixed $\frac{m}{n}$}\label{sec:experiments:FN}
We conduct a group of experiments to evaluate the performance of the proposed algorithms on synthetic graphs with a fixed average node degree. In the experiments, we vary $n$ from $5$ million to $2$ billion and let $\frac{m}{n}$ be $20$, as shown in Figure~\ref{fig:rangeV_S}. According to the experimental results drawn in Figure~\ref{fig:rangeV_S}, the I/O cost of each tested algorithm increases, and the increase trend of the I/O consumption of EP-BFS is linear. EB-BFS is timed out when $n=1$ billion, $n=1.5$ billion, and $n=2$ billion. The growth trends of both the CPU cost of EP-BFS and the I/O cost of EP-BFS are slow compared with the CPU and I/O consumption of EB-BFS. Based on Figure~\ref{fig:rangeV_S}(a), the running time of EP-BFS on the generated graph with $n=1.5$ billion is approximately equal to that of EB-BFS with $n=0.5$ billion. These experimental results prove that EP-BFS is an efficient semi-external algorithm that has good scalability when the number of nodes of the input graph increases.


\subsection{The impact of varying storage device}

In this part, the proposed algorithms are tested on both HDD and SSD with three groups of experiments, as shown in Figure~\ref{fig:rangeD_S} and Figure~\ref{fig:rangeV_S}. The datasets utilized in experiments shown in Figure~\ref{fig:rangeD_S} are introduced in Section~\ref{sec:experiments:AND}, while the others are introduced in  Section~\ref{sec:experiments:FN}. The experimental results depicted in Figures~\ref{fig:rangeD_S}-\ref{fig:rangeV_S} demonstrate that the proposed algorithms exhibit similar performance on two different storage devices. The reason is that all the proposed algorithms access the disk-resident dataset sequentially, which is fast. For instance,  on the graphs with $50$ million nodes, the time costs of EP-BFS on HDD and SSD are almost the same. 

These experimental results confirm that even though  currently SSDs are very fast, especially for applications that read data from them, utilizing such devices to support EP-BFS is not necessary. On the one hand, HDDs are sufficiently efficient for EP-BFS, for example the HDD used in our experiment supports USB 3.0, in that the data that EP-BFS reads from the HDD are stored continuously in its disk block, including (\romannumeral1) the input $G$ and (\romannumeral2) the data that EP-BFS writes into it. EP-BFS is designed specifically for HDDs, which never access data from the disk randomly. On the other hand, from our experiments, EP-BFS contains numerous CPU computations, which cause EP-BFS to still not be able to take full advantage of the read rate of current HDDs. More importantly, HDDs are more suitable for analyzing large-scale graphs under semi-external memory model in the industrial field, not only because HDDs are cheaper, but also because the fact that SSDs cannot support a large amount of data being written for a long time.

\subsection{The impact of varying $n$ on synthetic graphs with fixed $m$}
We also conduct experiments on a set of synthetic graphs $G$, where the edge number of $G$ is fixed at $5$ billion, and the node number of $G$ ranges from $50$ million to $250$ million, as shown in Figure~\ref{fig:rangeM_S}. Experimental results show that even though on the graphs with a fixed size, the time and I/O consumption of EP-BFS and that of EB-BFS increase with the increase in the node number of the generated graph. 

For instance, the time cost of EP-BFS on the generated graph with $n=50$ million is about $1$ hour, whereas the time cost of EP-BFS on the generated graph with $n=250$ million is higher than $2$ hours. EB-BFS consumes about $4$ hours and $7.5$ hours on the generated graph with $50$ million nodes and that with $250$ million nodes, respectively. Compared to the growth trend of the time consumption of EP-BFS or EB-BFS with the increase in the node number contained in the generated graphs, the growth trend of its I/O cost is slower. That is because, when $G$ contains more nodes, the time costs of the IMPs of both EP-BFS and EB-BFS increase. Besides, compared with EB-BFS, EP-BFS is more efficient. The growth trends of both the time and the I/O costs of EP-BFS are slower than those of EB-BFS. The reason is that EP-BFS does not maintain the entire spanning tree $T$ of $G$ in the main memory, and EP-BFS could reduce edges when enlarging its in-memory sketch, or at the end of certain iterations.

\begin{figure}[t]
	\centering
	\renewcommand{\thesubfigure}{}
	\subfigure[(a) Efficiency~(uk-2007-05)]{
		\includegraphics[scale = 0.85]{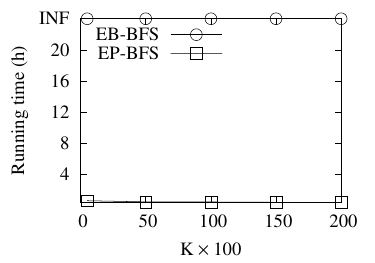}}
	\subfigure[(b) I/O~(uk-2007-05)]{
		\includegraphics[scale = 0.85]{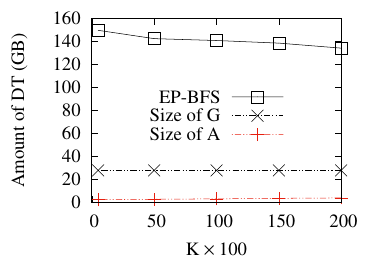}}
	\subfigure[(c) Efficiency~(Friendster)]{
		\includegraphics[scale = 0.85]{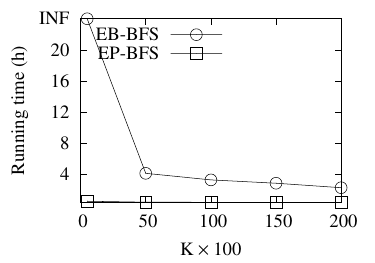}}
	\subfigure[(d) I/O~(Friendster)]{
		\includegraphics[scale = 0.85]{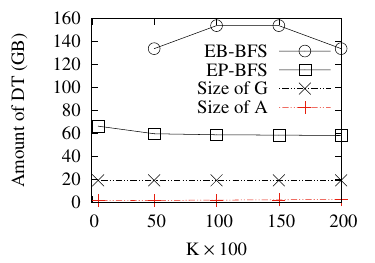}}
	\subfigure[(e) Efficiency~(synthetic graph)]{
		\includegraphics[scale = 0.85]{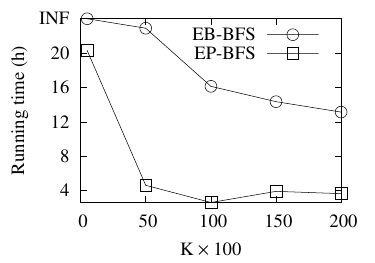}}
	\subfigure[(f) I/O~(synthetic graph)]{
		\includegraphics[scale = 0.85]{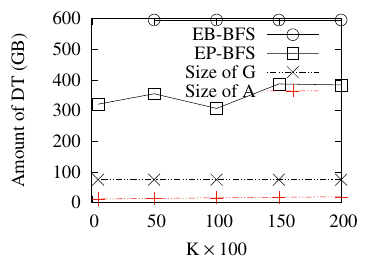}}
	\caption{The experimental results of varying memory size~(i.e. the size of $K$) on both real and synthetic graphs, where the utilized real graphs are  uk-2007-05 and Friendster, and the utilized synthetic graph $G$ contains $500$ million nodes, with the average node degree of $G$ being $20$.}
	\label{fig:rangeK_R}
\end{figure}

\subsection{The impact of varying $K$ on both real and synthetic graphs}

To evaluate the scalability of the proposed algorithms under different memory space sizes, we conduct experiments on two real datasets, namely, the graph uk-2007-05 and the graph Friendster, as well as a synthetic graph with $500$ million nodes and $10$ billion edges. We vary the value of $K$ in the range of $5\%$ to $200\%$, as shown in Figure~\ref{fig:rangeK_R}. The experimental results demonstrate that the scalability of EP-BFS under different memory spaces is significantly better than that of EB-BFS. Furthermore, when the available memory space of the computing device is small~(for example, $K=0.05$), EP-BFS still  performs well. 

Specifically, on the graph uk-2007-05, EB-BFS is timed out even when $3\times n$ edges can reside in the main memory. The performance of EP-BFS remains stable when the value of $K$ is varied on the graph uk-2007-05, which takes about half an hour and less than $150$GB I/Os. On the graph Friendster, EB-BFS times out when $K=5\%$, and requires about $3$ hours and $150$GB I/Os in the other cases, whereas EP-BFS exhibits good and stable performance, whose time cost is about $20$ minutes and I/O cost is about $58$GB. On the synthetic graph $G$, the performance of EP-BFS is also significantly better than that of EB-BFS. For instance, EB-BFS is timed out when $K=0.05$, while EP-BFS can process $G$ and obtain BFS results under our time limit. The CPU consumption of EP-BFS when $K=0.5$ is about $4$ hours, whereas the CPU consumption of EB-BFS when $K=0.5$ nearly reaches our time limit. The I/O cost of EP-BFS is stable which is about $350$GB, while the I/O cost of EB-BFS is about $600$GB when $K$ is ranged from $0.5$ to $2$.

\section{Conclusion} \label{sec:conclusion}
This paper considers the semi-external BFS problem for general disk-resident directed graphs, which is widely used in many applications. As semi-external BFS problem is still an open issue, this paper first proposes two naive algorithms with small MMSRs, named EE-BFS and EB-BFS, based on the basic framework of processing graph problems in semi-external memory model. Then, for efficiently processing current large-scale graphs with memory constraint, this paper proposes EP-BFS, the efficient semi-external BFS algorithm. EP-BFS is efficient because it can reduce the main costs of obtaining BFS results in semi-external memory model by (\romannumeral1) only maintaining part of $T$ in the main memory, (\romannumeral2) introducing a threshold to control the in-memory sketch enlarging process, (\romannumeral3) fully utilizing cache of the computing devices, etc. After presenting the correctness and complexity proofs of EP-BFS, experiments on both real and synthetic graphs are conducted to evaluate the performance of the proposed algorithms. Experimental results show that EE-BFS is inefficient, EB-BFS is relatively efficient and can be utilized in certain cases when $G$ contains millions of nodes or edges, whereas EP-BFS is efficient on large-scale graphs with over $1.7$ billion nodes or $91$ billion edges. 

\section*{Acknowledgments}
This work was supported in part by NSFC 62402135, U21A20513, 62202277, Taishan Scholars Program of Shandong Province grant tsqn202211091, and Natural Science Foundation of Shandong Province grant ZR2023QF059.




%
%
%
\end{document}